\documentclass[reqno,oneside,12pt]{amsart}

\newcommand{\DRAFT}{0}

\title[Self-consistent Coulomb interactions]{Self-consistent Coulomb interactions for \\ machine learning interatomic potentials}
\author{Jack Thomas}
\author{Will Baldwin}
\author{Gabor Csanyi}
\author{Christoph Ortner}
\date{\today}

\newcommand{\theclassifications}{\texttt{65E05}; \texttt{74E15}; \texttt{81V45}; \texttt{81V70}.}

\newcommand{\thekeywords}{Coulomb interactions; machine learning; electronic structure; tight binding; interatomic potentials; locality; body-order expansion.}

\newcommand{\theaddresses}{\texttt{jack.thomas@universite-paris-saclay.fr.} Laboratoire de Mathématiques d’Orsay, Université Paris--Saclay, CNRS, 91405, Orsay, France.\\
\texttt{wjb48@cam.ac.uk}, \texttt{gc121@cam.ac.uk}. Engineering Laboratory, University of Cambridge, Trumpington Street, Cambridge, CB2 1PZ, United Kingdom  \\ 
\texttt{ortner@math.ubc.ca}. Department of Mathematics, University of British Columbia, Vancouver, Canada.}
    
\usepackage{ifthen, xcolor,color, bm, braket, amsmath, amsthm, amssymb, mathtools, mathrsfs, microtype, enumitem, cite, esint}
\usepackage{soul}
\usepackage[hang,flushmargin]{footmisc}
\ifthenelse{\equal{\DRAFT}{1}}{\usepackage[notcite]{showkeys}}{}
\usepackage[colorlinks=true, allcolors=blue]{hyperref}
\usepackage{cleveref} 
\usepackage{datetime}

\renewcommand{\leq}{\leqslant}
\renewcommand{\geq}{\geqslant}
\renewcommand{\above}[2]{\genfrac{}{}{0pt}{}{#1}{#2}}
\DeclareMathOperator*{\argmin}{arg\,min}

\DeclareMathOperator*{\tr}{tr}
\DeclareMathOperator*{\Tr}{Tr}
\DeclareMathOperator*{\len}{len}
\DeclareMathOperator*{\dist}{dist}

\newcommand{\ep}{\varepsilon}
\newcommand{\Ham}{\mathcal{H}}
\newcommand{\vb}{\Xi}
\DeclareMathOperator*{\conv}{conv}

\usepackage{geometry}
\ifthenelse{\equal{\DRAFT}{1}}{\geometry{
    a4paper, left=5mm, right=45mm, top=20mm, bottom=25mm
}}{\geometry{
    a4paper, left=25mm, right=25mm, top=20mm, bottom=25mm
}}
\newtheorem{theorem}{Theorem}
\newtheorem{lemma}[theorem]{Lemma}
\newtheorem{prop}[theorem]{Proposition}
\newtheorem{corollary}[theorem]{Corollary}
\theoremstyle{remark}\newtheorem{rem}{Remark}\newtheorem{as}{Assumption}
\newenvironment{remark}{\begin{rem}}{\hfill $\star$ \end{rem}}

\numberwithin{equation}{section}
\crefformat{equation}{\textup{(#2#1#3)}}
\crefrangeformat{equation}{(#3#1#4)--(#5#2#6)}
\crefmultiformat{equation}{(#2#1#3)}%
{ and~(#2#1#3)}{, (#2#1#3)}{ and~(#2#1#3)}
\crefformat{section}{\S#2#1#3} 
\crefformat{subsection}{\S#2#1#3}
\crefformat{subsubsection}{\S#2#1#3}
\crefformat{prop}{Proposition~#2#1#3}
\crefformat{theorem}{Theorem~#2#1#3}
\crefformat{lemma}{Lemma~#2#1#3}
\crefformat{corollary}{Corollary~#2#1#3}

\newcounter{listcounter}
\newwrite\tempfile
\newcommand{\addcomment}[1]{
    \immediate\write\tempfile{\unexpanded{\stepcounter{listcounter}\texttt{\hyperlink{\arabic{listcounter}}{\arabic{listcounter}}:} #1 \newline}}
}
\newcounter{listcounter2}
\newcommand{\comment}[4][1]{\ifthenelse{\equal{#1}{1}}{\stepcounter{listcounter2}\hypertarget{\arabic{listcounter2}}{}\addcomment{#2: #4}}{}\ifthenelse{\equal{\DRAFT}{1}}{\textcolor{#3}{\texttt{\textup{[#4]}}}}{}}

\let\oldcite\cite
\renewcommand{\cite}[2][]{\oldcite{#2}\ifthenelse{\equal{\DRAFT}{1}}{\marginpar{\textcolor{blue}{\texttt{\tiny \ifthenelse{\equal{#2}{}}{\jt{\ifthenelse{\equal{#2}{}}{\href{#1}{MISSING CITATION}}{MISSING CITATION}}}{\ifthenelse{\equal{#1}{}}{#2}{\href{#1}{#2}}}}}}}{}}

\let\oldcref\cref
\renewcommand{\cref}[1]{\ifthenelse{\equal{#1}{}}{\marginpar{\texttt{\tiny{\jt{MISSING REFERENCE}}}}}{\oldcref{#1}\ifthenelse{\equal{\DRAFT}{1}}{\marginpar{\textcolor{gray}{\texttt{\tiny{#1}}}}}{}}}

\usepackage{newunicodechar}
\newunicodechar{✓}{\checkmark}

\begin{document}

\setcounter{page}{1}

\immediate\openout\tempfile=lists.tex
\immediate\write\tempfile{\noindent}

\begin{abstract}
   A ubiquitous approach to obtain transferable machine learning-based models of potential energy surfaces for atomistic systems is to decompose the total energy into a sum of local atom-centred contributions. 
   However, in many systems non-negligible long-range electrostatic effects must be taken into account as well. 
   We introduce a general mathematical framework to study how such long-range effects can be included in a way that (i) allows charge equilibration and (ii) retains the locality of the learnable atom-centred contributions to ensure transferability. 
   Our results give partial explanations for the success of existing machine learned potentials  that include equilibriation and provide perspectives how to design such schemes in a systematic way. 
   To complement the rigorous theoretical results, we describe a practical scheme for fitting the energy and electron density of water clusters.
\end{abstract}

\let\thefootnote\relax\footnote{
    \theaddresses\\
    \textit{2020 Mathematics Subject Classification:} \theclassifications
}


\maketitle

\thispagestyle{empty}

\newcommand{\ifDRAFT}[1]{\ifthenelse{\equal{\DRAFT}{1}}{#1}{}}
\newcommand{\jt}[2][1]{\comment[#1]{JT}{magenta}{#2}}
\newcommand{\jtm}[2][0]{\marginpar{\tiny{\comment[#1]{JT}{magenta}{#2}}}}
\newcommand{\jtedit}[1]{\textcolor{magenta}{#1}}

\newcommand{\co}[1]{{\color{red} #1}}
\newcommand{\cco}[1]{{\small \color{purple} \tt [CO: #1]}}

\newcommand{\wb}[1]{{\color{olive} \tt [WB: #1]}}

\section{Introduction}


\jtm[0]{background} Electronic structure models are widely used to predict optical, magnetic, and mechanical properties of materials and molecules. Today, \textit{ab initio} methods, such as density functional theory (DFT) \cite{ParrYang1994,bk:finnis,HohenbergKohn1964,KohnSham1965}, are too computationally expensive for large-scale simulations (but are still a popular choice for the simulation of systems up to a few hundred atoms, for example), whereas empirical force fields remain useful for large system sizes and long timescales. 
The introduction of machine-learning (ML) methodology into this field offers the prospect of bridging the gap between \textit{ab initio} and empirical models in order to derive models with \textit{ab initio} accuracy but at a fraction of the computational cost, enabling a systematic extension of the predictive first principles approach beyond the electronic structure length-scale to which it has been limited up until recently \cite{Bartok2017,Musil2021,Butler2018,Drautz2019-er,Batzner2022:nequip,MACE2022}.

This is often motivated by invoking the nearsightedness principle of electronic matter (NEM) \cite{Prodan_nearsightedness}. NEM concerns an electron density which is the ground state of an external potential $v$, given a fixed chemical potential. The statement is that if $v$ is changed in some region $\Omega$, then the response of the electron density at point $x$ decays towards zero as $x$ moves away from $\Omega$. NEM therefore suggests that local machine learning models are effective so long as a change in geometry does not induce changes in \textit{potential} at some distant point. This is not the case when, for instance, a reorientation of a polar molecule leads to a change in external potential even at distant points. 
In order to account for electrostatic effects, a long-range pairwise term can be added to the standard machine learning energy contribution that maps local geometry to energy \cite{Bartok2010GaussianElectrons}. However, changes in the chemical environment may induce changes in the charge distribution at long-range (even if the local geometry is unchanged). It is therefore important in many systems to include electronic information directly into the machine learning framework. In this paper, we derive a mathematically rigorous scheme for including enhanced electronic information into machine learned interatomic potentials. In doing so, we go some way to justifying the ML charge equilibriation schemes which have been proposed in the literature, including the Fourth generation neural network potential \cite{Ko2021}, the Becke population neural network (BpopNN) \cite{Xie2020} and the self consistent field neural network (SCFNN) \cite{Gao2022}. 



\normalsize

\jtm[0]{locality in ML schemes}
To efficiently parameterise the complex many-body Born--Oppenheimer potential energy surface (PES) using machine-learning methodology, one must decompose it into lower dimensional components. In previous works, we established two groups of results of this kind, both for simplified models that explicitly replaced Coulomb with an exponentially localized Yukawa interaction: The total potential energy can be decomposed into a sum of atom-centred site energy contributions that may be chosen (with controllable error) to \textit{(i)} depend only on atoms within a finite cut-off radius $r_{\mathrm{cut}}$\cite{ChenOrtnerThomas2019:locality,Thomas2020:scTB}, and (\textit{ii}) have finite \textit{correlation order} $N$ \cite{ThomasChenOrtner2022:body-order}. The error committed in this approximation is exponentially small in both $r_{\mathrm{cut}}$ and $N$. 
That is, for an atomic system
$\bm r = \{\bm r_\ell\}_{\ell=1}^M \subset \mathbb R^d$, 
we may approximate the total energy $E(\bm r) = \sum_{\ell} E_\ell(\bm r)$ by: 
\begin{align}
    %
    %
    %
    %
    \label{eq:body-order-intro}
    E_{\ell}(\bm r) &\approx 
    V_0 + 
    \sum_{\above{k\not=\ell:}{r_{\ell k} < r_{\mathrm{cut}}}} V_1(\bm r_{\ell k}) + 
    \sum_{
    \above
        {k_1,k_2\not=\ell\colon}
        {r_{\ell k_i}<r_{\mathrm{cut}} \forall i}
    } 
    V_2(\bm r_{\ell k_1}, \bm r_{\ell k_2}) 
    + \cdots + 
    \sum_{
    \above
        {k_1,\dots,k_N \not=\ell\colon}
        {r_{\ell k_i}<r_\mathrm{cut} \forall i}
    } 
    V_N(\bm r_{\ell k_1},\dots, {\bm r}_{\ell k_N}), 
    %
    %
\end{align}
where $\bm r_{\ell k} \coloneqq \bm r_k - \bm r_\ell$
, $r_{\ell k} \coloneqq |\bm r_{\ell k}|$, and $V_n(\bm r_{\ell k_1},\dots,\bm r_{\ell k_n})$ is an $(n+1)$-\textit{body} (or $n$-\textit{correlation}) \textit{potential} modelling the interaction of a central atom $\ell$ and the $n$ neighbouring atoms $k_1,\dots,k_n$. 

Exponential convergence of \cref{eq:body-order-intro} in terms of both cut-off radius and correlation order~\cite{ChenOrtnerThomas2019:locality,Thomas2020:scTB, ThomasChenOrtner2022:body-order}
lends theoretical support to 
the empirical success of body-ordered approximations in popular machine learning schemes \cite{Drautz2019-er,Shapeev2016}. 
%
However, an underlying assumption in all of these results is the total screening of Coulomb interactions. While an appropriate simplification for a variety of systems, this assumption is a significant limitation of the results, and also considerably simplifies the mathematical analysis. 

\jtm[0]{ML + explicit electronic information}
In response, some machine learning schemes began to include an additional long-range electrostatic contribution to the total energy of the system. 
For example, in \cite{Artrith2011,Yao2018-ChemSci9}, the authors incorporate long-range interactions by training additional neural networks to produce atomic charges that either predict the \textit{ab initio} point charges \cite{Artrith2011} directly, or do so by learning the dipole moments \cite{Yao2018-ChemSci9}. 
However, these methods learn the long-range electrostatics from machine learning schemes with local descriptors and therefore cannot be expected to correctly adapt to changes in the chemical environment that result in non-local charge equilibration. 

\jtm[0]{LODE} 
In a different direction, the long-distance equivariant (LODE) framework \cite{Grisafi2019:LODE} builds on SOAP descriptors \cite{Bartok2013:SOAP} by incorporating the electrostatic potential corresponding to proxy densities to provide enhanced electronic information directly into the descriptors. These additional non-local descriptors are fixed, and are again not equilibriated in a way resembling the underlying physics.  

\jtm[0]{charge equilibriation } An alternative perspective is to revisit the classical electronegativity equalisation method (EEM) \cite{Mortier1985,Mortier1986} (or, charge equilibriation (QEq) \cite{Rappe1991-PhysChem95-QEq,Rick1994}), as well as improvements based on the same idea \cite{Chen2007-ChemPhysLett438-QTPIE, Wilmer2012-PhysChemLett3-EQeq}. One postulates an extended potential energy surface (PES) as a function of partial charges $q_\ell \coloneqq Z_\ell - p_\ell$ where $Z_\ell$ is the atomic species and $p_\ell$ is an electron population on atom $\ell$. This extended PES is minimised with respect to the charges, together with constraints on the total charge in the system, to produce a charge-equilibrated PES: 
\begin{align}
    \label{eq:QEq}
    E(\bm r) = \min_q E(\bm r, q)
    \quad \text{where} \quad 
    E(\bm r, q) &\coloneqq 
    \sum_{\ell} 
    \big[
        \ep_\ell + \chi_\ell q_\ell + U_\ell q_\ell^2
    \big]
    + \sum_{\ell<k} J_{\ell k} q_\ell q_k.
\end{align}
The first term of $E(\bm r, q)$ is the intra-atomic energy resulting from a Taylor series expansion of the individual atomic energies of the constituent elements with $\chi$ the electronegativity and $U$ the atomic hardness. 
These quantities may be written in terms of atomic ionisation potentials (IP) and electron affinities (EA), which traditionally were obtained from experimental data. 
The charges at different atomic positions interact through the Coulomb potential $J$. Since $E(\bm r, q)$ is a quadratic function of $q$,  minimising with respect to the charges, leads to a linear system of equations which may be solved subject to a constraint on the total charge of the system.
%

\jtm[0]{QEq + ML}
In reality, the parameters in \cref{eq:QEq} should be environment-dependent: 
$\chi_\ell = \chi_\ell(\bm r)$ and $U_{\ell} = U_\ell(\bm r)$. 
Exploiting the fact that the atomic electronegativities are local functions of the environment \cite{Ghasemi2015}, a natural approach is to parameterise them using a machine learning framework, leading to the charge equilibration via neural network technique (CENT) \cite{Ghasemi2015,Faraji2017}. The basic idea in this method can be represented in the following schematic:
\begin{align}
    \label{eq:Ghasemi}
    \{\bm r_{\ell k}\}_{k\not= \ell} 
    \mapsto_{\text{local, ML}} \chi_\ell(\bm r)
    \qquad \text{and} \qquad
    \chi \mapsto_{\textrm{linear}} q
    \mapsto_{\textrm{explicit}} E, \nabla E.
\end{align}
The first mapping in \cref{eq:Ghasemi} is approximated through a neural network architecture, whereas the simple quadratic form of \cref{eq:QEq} means the mapping from electronegativities to energies $E$ and forces $\nabla E$ can be efficiently and cheaply evaluated through solving a linear system of equations and evaluating explicit functions. This approach therefore builds on the simplicity of QEq schemes with a machine learning approach, leading to improved accuracy and transferrability. However, by its very nature, a quadratic-in-$q$ ansatz is too simple to fully model the underlying physics. For example, this scheme can only describe systems where the ground state electron density (and potential energy) is a smooth function of the geometry. This is not the case for systems which exhibit conical intersections - wherein the crossing in energies between two diabatic states leads to a discontinuous electron density as a function of geometry. Such conical intersections already occur in simple molecules\cite{Vieuxmaire_conical_intersection, Perun_CI_thymine}.


\jtm[0]{XPS}
In the Becke population neural network (BpopNN) \cite{Xie2020}, the electronegativity and atomic hardness are fixed, and an additional neural network correction is added to the total energy from \cref{eq:QEq} that depends on both the atomic positions and charges,  
\begin{align}
    \label{eq:XPS}    
    E(\bm r, q) 
    = \sum_{\ell} f_{Z_\ell}(\bm r, q)
    + \sum_{\ell} \big[
        \ep_\ell + \chi_{Z_\ell} q_\ell + U_{Z_\ell} q_\ell^2
    \big]
    + \sum_{\ell<k} J_{\ell k} q_\ell q_k
\end{align}
where $f_{Z_\ell}$ is an atomic neural network for the atomic species $Z_\ell$. Then, in contrast to \cref{eq:QEq}, the entire extended PES is approximated by targeting an \textit{ab initio} PES of the form 
$\min_{\Psi \to q} \Braket{\Psi | \Ham | \Psi}$ 
where $\Psi$ is an admissible wavefunction in this minimisation problem if the corresponding electron density $\rho_{\Psi}$ satisfies 
$q_\ell = Z_\ell - \int \rho_\Psi(x) \omega(x - \bm r_\ell) \mathrm{d}x$ 
for some chosen weight function $\omega$ concentrated around the origin. In practice, this is done using constrained DFT implemented in Q-Chem \cite{Shao2015:QChem}. Minimising the extended PES with respect to charges mimics the charge-equilibration in DFT models, and thus the electronic contributions are able to self-consistently adapt to the global environment. 

Thus, the overarching idea is to consider an extended PES as a function of the atomic configuration together with additional atom-centred features which provide a low-dimensional representation of the electron density (in BpopNN, the additional features are Becke populations). This model energy may then be minimised with respect to the additional features to provide a self-consistent PES. 
\jtm[0]{What we actually prove}
In this paper, we rigorously justify this approach by introducing atom-centred features $\widehat{v}$ describing the effective potential: 
%
\begin{align}
    \label{eq:extended-PES}
    E(\bm r, \widehat{v}) &= 
    \sum_{\ell} E_\ell(\bm r, \widehat{v}) +
    E_\mathrm{el}\big[\rho(\bm r, \widehat{v})\big] 
    \qquad \text{where} \qquad
    \rho(\bm r, \widehat{v}) = \sum_\ell \rho_\ell(\bm r, \widehat{v}) 
\end{align}
and $E_\mathrm{el}$ is an explicit long-range function of the electron density. 
We then show that the quantities 
$E_\ell, \rho_\ell$ 
are local functions of the extended variables 
$\{\bm r_{k}, \widehat{v}_k\}_k$
and can thus be approximated using a machine learning framework. In particular, we are able to approximate \cref{eq:extended-PES} with a body-ordered PES 
$E_N(\bm r, \widehat{v})$ 
analogous to \cref{eq:body-order-intro}, and show that  minimisers $\widehat{v}^\star$ to \cref{eq:extended-PES} (and corresponding energies) may be approximated with minimisers $\widehat{v}^\star_N$ to $E_N$ (and the corresponding approximate energies). The scheme is systematically improvable by choosing higher body-ordered approximations.


While this approach is convenient to obtain theoretical results that partially justify the BpopNN approach \cite{Xie2020}, we also present numerical experiments for a closely related but computationally more practical fixed point scheme. We build a machine learning based surrogate model describing $E_\ell$ and $\rho_\ell$ as functions of the extended variables $\{\bm r_k, \widehat{v}_k\}_k$. Then, instead of minimising the approximate energy as in the theoretical results, in practice, it is more convenient to approach the critical points of the energy by iteratively solving the corresponding Euler--Lagrange equation directly. We show convergence plots for water clusters as the number of features on each atomic site is increased.

The paper is organised as follows. In \cref{sec:el-struct}, we discuss popular Kohn--Sham models (\cref{sec:KS}) from which we derive corresponding tight-binding (\cref{sec:TB}) and machine learning  schemes (\cref{sec:ML}). In \cref{sec:results} we state that these resulting machine learning models can be decomposed as in \cref{eq:extended-PES}. Numerical experiments are presented in \cref{sec:numerics}. Proofs of the main results are collected together in \cref{sec:proofs}. A brief summary of the main notation used throughout this paper is contained in \Cref{sec:notation}.

\section{Electronic Structure Models}
\label{sec:el-struct}
In this section, we consider a sequence of approximations starting with Kohn--Sham density functional theory and resulting in machine learning schemes that provide a convenient framework for the theoretical results of this paper. When discussing these electronic structure models, we follow the majority of numerical studies and consider the \textit{canonical ensemble} for the electrons; the number of particles in the system, the volume, and Fermi-temperature are all fixed. 
However, since the rigorous results of this paper are stated and proved in the \textit{grand-canonical ensemble} (where the  chemical potential, volume, and Fermi-temperature are fixed), in Remarks~\ref{rem:GCE1}, \ref{rem:finite-temp-GCE-TB} and \ref{rem:GCE} we explain how these models are adapted to this setting. In \S\ref{sec:results} (Remark~\ref{rem:CE}) we explain how, in principle, one may extend the analysis of this paper to the canonical ensemble; a possible direction for future work.

\subsection{Kohn--Sham models} \label{sec:KS}
For simplicity of notation, we shall consider spin unpolarised systems with an even number of electrons which allows us to omit the spin variable in the following. We consider a system of $M$ atoms at positions 
$\bm r_1, \dots, \bm r_M \in \mathbb R^3$ 
with corresponding atomic charges $Z_1,\dots,Z_M$, and $N_{\mathrm{el}}$ electron pairs.
The corresponding \textit{one-particle density operator} of this system is a self-adjoint projector $\gamma$ on $L^2(\mathbb R^3)$ with $\mathrm{Tr}\, \gamma = N_{\mathrm{el}}$.
Since $\gamma$ is compact (as it is trace-class), $\gamma$ may be diagonalised in an orthonormal basis $\varphi_i = \ket{i}$, leading to the following expressions for $\gamma$ and the corresponding electron density $\rho_\gamma$: 
\[
    \gamma = \sum_{i=1}^{N_{\mathrm{el}}} \ket{i}\bra{i}, \qquad \rho_\gamma(x) \coloneqq 2\sum_{i=1}^{N_{\mathrm{el}}} |\varphi_i(x)|^2
\]
where $\bra{i}\in L^2(\mathbb R^3)^\star$ with $\bra{i} \phi \coloneqq (\varphi_i, \phi)_{L^2(\mathbb R^3)}$ and $\braket{i|j} \coloneqq (\varphi_i,\varphi_j)_{L^2(\mathbb R^3)} = \delta_{ij}$. The density clearly satisfies $\rho_\gamma \geq 0$ and $\int \rho_\gamma = 2N_{\mathrm{el}}$. 

The total energy of the system, in Kohn--Sham models, takes the form \cite{Anantharaman2009}
\begin{align}
    \label{eq:ab-initio}
    \mathcal E(\gamma) = 
    \mathrm{Tr} \big[ -\tfrac{1}{2}\Delta  \gamma\big]
    + \int \rho_\gamma V^\mathrm{nuc} 
    + \frac{1}{2}\iint \frac{\rho_\gamma(x)\rho_\gamma(y)}{|x-y|}\mathrm{d}x\mathrm{d}y  
    + E_\mathrm{xc}[\gamma]
\end{align}
where 
$V^\mathrm{nuc}(x) \coloneqq - \sum_i \frac{Z_i}{|x - \bm r_i|}$ 
is the potential generated by the nuclei and the kinetic energy is given by 
\[
    \mathrm{Tr} \big[ -\tfrac{1}{2}\Delta \gamma \big] \coloneqq 
    \frac{1}{2} \sum_{i=1}^{N_{\mathrm{el}}} \|\nabla \varphi_i\|_{L^2(\mathbb R^3)}^2 \in [0,\infty].
\]
Therefore, if $\gamma$ has finite kinetic energy, then 
$\sqrt{\rho_\gamma} \in H^1(\mathbb R^3)$ 
and the Sobolev embedding \cite{Aubin1982} implies
$\rho_\gamma \in L^1 \cap L^3(\mathbb R^3) \subset L^{6/5}(\mathbb{R}^3)$. 
The latter fact also ensures the Coulomb energy 
$\iint \frac
    {\rho_\gamma(x)\rho_\gamma(y)}
    {|x-y|}
\mathrm{d}x\mathrm{d}y$ 
is finite since $\iint \frac{f(x) g(y)}{|x-y|} \mathrm{d}x \mathrm{d}y\lesssim \|f\|_{L^{6/5}(\mathbb R^3)} \|g\|_{L^{6/5}(\mathbb R^3)}$ (which follows from the Hardy--Littlewood--Sobolev inequality  \cite{Aubin1982}).

\begin{remark}[Smeared nuclei]
    \label{rem:smeared-nuclei}
    If we instead consider smeared nuclei (which is done implicitly in CENT \cite{Ghasemi2015,Faraji2017}) with distribution 
    $\nu(x) = \sum_{k=1}^M Z_k m_{k}(x - \bm r_k)$ 
    where $m_1,\dots,m_M$ are smooth non-negative functions with $\int m_k = 1$, then 
    $V^\mathrm{nuc}(x) = -\int \frac{\nu(y)}{|x-y|}\mathrm{d}y$. 
    Therefore, up to a constant representing the nuclei-nuclei interaction, we have
    \begin{align}
        \label{eq:smeared-nuclei}
        \mathcal E(\gamma) 
        = \mathrm{Tr}\big[-\tfrac{1}{2}\Delta \gamma\big] 
        + E_{\mathrm{xc}}[\gamma]
        + \frac{1}{2}
        \iint 
            \frac
                {\big(\rho_\gamma(x)- \nu(x)\big)
                \big(\rho_\gamma(y)-\nu(y)\big)}
                {|x-y|} 
        \mathrm{d}x\mathrm{d}y.
    \end{align}
    The Coulomb term from \cite{Ghasemi2015} is $J_{\ell k} = \mathrm{erf}(\tau_{\ell k} r_{\ell k})$ and results from the final term in \cref{eq:smeared-nuclei} by assuming a Gaussian distribution for both the nuclei and the electron density.
\end{remark}

In theory, there is a universal exact exchange-correlation functional \cite{KohnSham1965, HohenbergKohn1964}. However, this function is complicated and unknown. In practice, $E_\mathrm{xc}$ is replaced with an approximation:
\begin{itemize}
    \item $E_\mathrm{xc} = 0$ : reduced Hartree-Fock (also known as the Hartree model \cite{Hartree1928}), 
    
    \item $E_\mathrm{xc}[\gamma] = - \frac{1}{2} \iint \frac{|\gamma(x,y)|^2}{|x-y|}\mathrm{d}x\mathrm{d}y$ : Hartree--Fock. This model results from a variational formulation restricted to the set of finite energy Slater determinants, 
    
    \item $E_\mathrm{xc}[\gamma] = \int \ep_\mathrm{xc}\big(\rho_\gamma(x)\big) \mathrm{d}x$ : local density approximation (LDA) \cite{KohnSham1965}. In this formulation, $\rho^{-1} \ep_\mathrm{xc}(\rho)$ is the exchange-correlation density for a uniform electron gas with density $\rho$. 
    The simplest such approximation is given by $\ep_\mathrm{xc}(\rho) \coloneqq - c_0 \rho^{\frac{4}{3}}$ where $c_0$ is a positive constant, 
    
    \item $E_\mathrm{xc}[\gamma] = 
    \int 
        \ep_\mathrm{xc}\big(
            \rho_\gamma(x), 
           |\nabla  \textstyle{\sqrt{\rho_\gamma(x)}}|^2
        \big) 
    \mathrm{d}x$ : the generalised gradient approximation (GGA) is the next simplest functional form that allows for the inhomogeneity of the electron density.
\end{itemize}
After choosing a level of theory, \cref{eq:ab-initio} is minimised over the space of finite energy (one particle) density operators as described above:
\begin{align}
    \mathcal P_{N_\mathrm{el}} \coloneqq 
    \left\{ 
        \gamma \in \mathcal S\big( L^2(\mathbb R^3) \big) \colon 
        \gamma^2 = \gamma, \,\, 
        \mathrm{Tr}\,\gamma = N_{\mathrm{el}}, \,\, 
        \mathrm{Tr} \big[ - \tfrac{1}{2} \Delta \gamma \big] < \infty 
    \right\}
    \label{eq:P}
\end{align}
where $\mathcal S\big( L^2(\mathbb R^3) \big)$ is the space of self-adjoint operators on $L^2(\mathbb R^3)$. In extended Kohn-Sham models, one minimises \cref{eq:ab-initio} over the convex hull of $\mathcal P_{N_\mathrm{el}}$, denoted 
\begin{align}
    \mathcal K_{N_\mathrm{el}} \coloneqq \mathrm{conv}(\mathcal P_{N_\mathrm{el}})
    = \left\{ 
        \gamma \in \mathcal S\big( L^2(\mathbb R^3) \big) \colon 
        0\leq \gamma \leq 1, \,\, 
        \mathrm{Tr}\,\gamma = N_{\mathrm{el}}, \,\, 
        \mathrm{Tr} \big[ - \tfrac{1}{2} \Delta \gamma \big] < \infty 
    \right\}
    \label{eq:K}
\end{align}
allowing for fractional occupation numbers. The existence of a minimiser has been established for the various levels of theory: Hartree model \cite{Solovej1991}, Hartree--Fock \cite{Lions1987,LiebSimon1977}, extended Hartree--Fock \cite{Lieb1983}, LDA \cite{LeBris1993}, and extended LDA and GGA \cite{Anantharaman2009}. 

We consider the minimisation problem \cref{eq:ab-initio} within the LDA theory (though conceptually at least our approach can be applied more generally). We first introduce mild assumptions on $\ep_\mathrm{xc}$ that are satisfied for practical models:  $\ep_{\mathrm{xc}} \colon \mathbb R_+ \to \mathbb R$ 
is twice differentiable with 
$\ep_{\mathrm{xc}} \sim - \rho^{4/3}$ near zero and 
$-\rho^{1/3} \lesssim \ep_\mathrm{xc}^\prime(\rho) \leq 0$ on $(0,\infty)$. 
Moreover, we define the Hamiltonian 
    \begin{align}
        \label{eq:LDA-Ham}
        \Ham[\rho] &\coloneqq -\frac{1}{2}\Delta + V^\mathrm{nuc} + v_{\rm C} \rho + \ep_\mathrm{xc}^\prime(\rho).
    \end{align}
where $v_{\rm C}f(x) \coloneqq \int \frac{f(y)}{|x - y|}\mathrm{d}y$ denotes the Coulomb potential. For a given $\rho$, there exists an associated \textit{Fermi level}  $\ep_{\mathrm{F}}$ satisfying $\Tr \chi_{(-\infty,\ep_{\mathrm{F}}]}\big(\Ham[\rho]\big) = N_{\mathrm{el}}$.
We have the following existence result \cite{Anantharaman2009}:

\begin{theorem}
    For neutral or positively charged systems, there exists a minimiser $\gamma^0$ of \cref{eq:ab-initio} over $\mathcal K_{N_{\mathrm{el}}}$. If there exists a Fermi level associated to $\rho_{\gamma^0}$ such that $\ep_{\mathrm{F}} \not\in \sigma\big(\Ham[\rho_{\gamma^0}]\big)$, then
    \begin{align}
        \label{eq:LDA}
        \gamma^0 &= \chi_{(-\infty,\ep_\mathrm{F})}\big( \Ham[\rho_{\gamma^0}]\big).
    \end{align}
    
\end{theorem}

\begin{remark}[Finite Fermi-temperature]
    \label{rem:finite-temp}
    The main results of this work also apply to the finite Fermi-temperature setting \cite{Mermin1965}. In this case, the total energy of the system is
    \begin{align}
        \label{eq:ab-initio-beta}
        \mathcal E_\beta(\gamma) \coloneqq \mathcal E(\gamma) + \beta^{-1} \mathrm{Tr}\, S(\gamma) 
    \end{align}
    where $S(x) \coloneqq x\log x + (1-x)\log(1-x)$ is the Fermi-Dirac entropy and $\beta = (k_\mathrm{B} T)^{-1}$ is the \textit{inverse Fermi-temperature} where $k_{\mathrm{B}}$ is the Boltzmann constant and $T$ the absolute temperature. (The zero Fermi-temperature case \cref{eq:ab-initio} is naturally included with $\mathcal E = \mathcal E_{\infty}$). Then, as in extended Kohn--Sham models, this finite Fermi-temperature energy is minimised over $\mathcal K_{N_\mathrm{el}}$, yielding
    \begin{align}
        \label{eq:EL-beta}
        \gamma^0 
        = F_{\ep_{\mathrm{F}},\beta}\big( 
            \Ham[\rho_{\gamma^0}]
        \big),
    \end{align}
    where $F_{\mu,\beta}(x) \coloneqq \big( 1 + e^{\beta(x - \mu)}\big)^{-1}$ is the Fermi-Dirac distribution with chemical potential $\mu$ and inverse Fermi-temperature $\beta$. 
    
    Assuming there is a gap in the spectrum at the Fermi-level, and using the notation 
    $F_{\ep_{\mathrm{F}},\infty} 
    \coloneqq \lim\limits_{\beta \to \infty} F_{\ep_{\mathrm{F}},\beta} 
    = \chi_{(-\infty, \ep_{\mathrm{F}})} + \tfrac{1}{2}\chi_{\{\ep_{\mathrm{F}}\}}$, 
    the Euler--Lagrange equation \cref{eq:EL-beta} at zero Fermi-temperature agrees with \cref{eq:LDA}.

    The main results of this paper (see next section) are stated and proved both for insulators at zero temperature and in the case of finite Fermi-temperature.
\end{remark}

\begin{remark}[Grand canonical ensemble]\label{rem:GCE1}
    For fixed chemical potential $\mu \in \mathbb R$, we also define the grand-canonical free energy:
    \begin{align}
        \label{eq:ab-initio-GCE}
        \mathcal G_{\beta,\mu}(\gamma) 
        \coloneqq \mathcal E_\beta(\gamma) - \mu \mathrm{Tr}\,\gamma
    \end{align}
    which is minimised over the set of one-particle density operators without the constraint on the number of electrons in the system: 
    $\mathcal P \coloneqq \bigcup_{N_{\mathrm{el}} \geq 0} \mathcal P_{N_\mathrm{el}}$ 
    (for standard Kohn--Sham models at zero Fermi-temperature) or 
    $\mathcal K \coloneqq \bigcup_{N_{\mathrm{el}} \geq 0} \mathcal K_{N_\mathrm{el}}$ 
    (for extended models at zero Fermi-temperature, or at finite Fermi-temperature). In this way, the chemical potential represents the contribution to the free energy due to varying particle number. 

    The main results of this paper are stated for the grand-canonical ensemble, and we comment further in Remark~\ref{rem:CE} (below) regarding extending the analysis to the canonical ensemble.
\end{remark}

For the remainder of this paper, it will be convenient to consider the \textit{(one-particle) density matrix}, the integral kernel of the density operator: for $\gamma \in \mathcal P_{N_\mathrm{el}}$, it is common to abuse notation and also write 
$\gamma \in L^2(\mathbb R^3 \times \mathbb R^3)$ 
such that
\begin{align}
    \label{eq:DM}
    (\gamma \psi)(x) = 
    \int \gamma(x,y) \psi(y) \mathrm{d}y, 
    \qquad \forall \psi \in L^2(\mathbb R^3).
\end{align}

\subsection{Tight Binding}
\label{sec:TB}
We expand the one-particle density operator 
$\gamma = \sum_{i} f_i \ket{i}\bra{i} 
\in \mathcal K$ 
in a local basis of atomic-like orbitals. That is, to each atomic site $\ell \in \{1,\dots,M\}$, we assign 
$N_\mathrm{b} = N_{\mathrm{b}}(Z_\ell)$ 
atomic-like orbitals 
$\ket{\ell a} \coloneqq 
\phi_{\ell a} \coloneqq 
\phi_{Z_\ell, a}(\,\cdot - \bm r_\ell)$
centred on $\bm r_\ell$ and depending on the atomic species $Z_\ell$ (for $a = 1, \dots, N_\mathrm{b}$). Then, on expanding $\ket{i} = \sum_{\ell a} C^i_{\ell a} \ket{\ell a}$, we have
\begin{align}
    \label{eq:DM-TB}
    \gamma(x,y) =
    \phi(x)^\mathrm{T} \gamma_{\mathrm{TB}} \phi(y) \coloneqq 
    \sum_{\ell k, ab} \phi_{\ell a}(x) [\gamma_{\mathrm{TB}}]_{\ell k, ab} \phi_{kb}(y),
\end{align}
where $\gamma_{\mathrm{TB}} \coloneqq \sum_{i} f_i C^i \otimes C^i$. 
Throughout this paper, we will assume that the orbitals are exponentially localised: there exists $c_\phi, \eta_\phi > 0$ such that 
\begin{align}
    \label{eq:phi-assumption}
    \big|\partial^\alpha \phi_{Z, a}(x)\big| \leq c_\phi e^{-\eta_\phi  |x|}
\end{align}
for all atomic species $Z$, basis indices $a = 1, \dots, N_{\mathrm{b}}$, and multi-indices $\alpha \in \mathbb N^3$ with $|\alpha|_1 \leq \nu$ for some $\nu \geq 2$. (We require $\nu = 2$ for the main results, but will comment on possible extensions if $\nu > 2$.)

For the remainder, we will simply write 
$\gamma = \phi^\mathrm{T} \gamma_{\mathrm{TB}} \phi$ 
both for the density matrix \cref{eq:DM-TB} as well as the corresponding density operator as defined by \cref{eq:DM}. Moreover, the electron density $\rho_\gamma$ corresponding to $\gamma$ is given by 
$\rho_\gamma(x) \coloneqq 2 \phi(x)^\mathrm{T} \gamma_{\mathrm{TB}} \phi(x)$. 
We suppose that 
$\braket{\ell a| kb} = \delta_{\ell k}\delta_{ab}$ 
that is, we assume an \textit{othogonal} tight binding model. 
This can be achieved in practice by considering a similarity transformation of the orbitals such as the L\"owdin transform \cite{ChenOrtner16,Benzi2013}. A consequence of this is that $C^i \cdot C^j = \delta_{ij}$. In the following, we also suppose that the number of orbitals per atom is constant. We do this without loss of generality \oldcite[Appendix~B]{OrtnerThomas2020:arXiv}.\jtm{not in \oldcite{OrtnerThomas2020:pointdefs}}

When 
$\gamma = \phi^\mathrm{T} \gamma_\mathrm{TB} \phi$ 
as in \cref{eq:DM-TB}, and assuming an LDA exchange correlation functional, the \textit{ab initio} energy \cref{eq:ab-initio} reduces to  
\begin{gather}
    \label{eq:TB-energy}
    \mathcal E(\gamma) = 
    \mathrm{Tr} \,H_0\gamma_{\mathrm{TB}}
    + \int \Big[ \rho_\gamma V^\mathrm{nuc} 
    + \tfrac{1}{2}\rho_\gamma v_{\rm C}\rho_\gamma 
    + \ep_{\mathrm{xc}}\big( \rho_\gamma \big) \Big] \\ 
    \label{eq:TB-H0}
    \text{where} \qquad [H_0]_{\ell k, ab} \coloneqq \frac{1}{2}\int \nabla \phi_{\ell a} \cdot \nabla \phi_{kb}.
\end{gather}
Therefore, we also suppose that 
$\{ \phi_{\ell a} \} \subset H^1(\mathbb R^3)$. 

In particular, on making the tight-binding approximation, we are restricting the class of one-particle density operators to $\mathcal P^{\mathrm{TB}}_{N_\mathrm{el}}$ where
\begin{align}
    \mathcal P^{\mathrm{TB}}_{N_\mathrm{el}} &\coloneqq 
    \left\{ 
        \phi^{\mathrm T} \gamma_{\mathrm{TB}} \phi \colon 
        \gamma_\mathrm{TB} \in (\mathbb R^{N_\mathrm{b}\times N_{\mathrm{b}}})^{M\times M}, \,\, 
        \gamma_{\mathrm{TB}}^{\mathrm T} = \gamma_{\mathrm{TB}}, \,\,
        \gamma_{\mathrm{TB}}^2=\gamma_{\mathrm{TB}}, \,\,
        \mathrm{Tr}\, \gamma_{\mathrm{TB}} = N_{\mathrm{el}}
    \right\} \nonumber \\
    &\subset \mathcal P_{N_{\mathrm{el}}}.
    \label{eq:P-TB}
\end{align}
We also define the projection of the Hamiltonian \cref{eq:LDA-Ham} onto the tight binding basis by
\begin{align}
    \label{eq:Ham[rho]}
    \Ham^{\mathrm{TB}}[\rho]_{\ell k, ab} &\coloneqq 
    \int \phi_{\ell a} \Big[ 
        -\frac{1}{2}\Delta + V^\mathrm{nuc} + v_{\rm C} \rho + \ep_\mathrm{xc}^\prime(\rho)
        \Big] \phi_{kb}, 
\end{align}
and write $\lambda_1(\rho) \leq \dots \leq \lambda_{M\cdot N_\mathrm{b}}(\rho)$ for the ordered eigenvalues of $\Ham^{\mathrm{TB}}[\rho]$ counting multiplicities.

The electronic structure problem in the tight binding framework then consists of minimising $\mathcal E$ \cref{eq:TB-energy} over $\mathcal P_{N_\mathrm{el}}^{\mathrm{TB}}$:  
\begin{prop}
    \label{TB-minimiser}
    There exists  
    $\gamma^0 = \phi^\mathrm{T} \gamma_\mathrm{TB}^0 \phi$, 
    a minimiser of $\mathcal E$ over $\mathcal P^{\mathrm{TB}}_{N_\mathrm{el}}$. If there is a gap in the spectrum (that is, $\lambda_{N_{\mathrm{el}}}(\rho_{\gamma^0}) < \lambda_{N_{\mathrm{el}}+1}(\rho_{\gamma^0})$), then  
    \begin{align}
        \gamma^0_\mathrm{TB} &= 
        \chi_{(-\infty,\ep_\mathrm{F})}\big( 
            \Ham^{\mathrm{TB}}[\rho_{\gamma^0}]
        \big) 
    \end{align}
    where the Fermi-level $\ep_\mathrm{F}$ is any value satisfying $\lambda_{N_{\mathrm{el}}}(\rho_{\gamma^0}) < \ep_{\mathrm{F}} < \lambda_{N_{\mathrm{el}}+1}(\rho_{\gamma^0})$. 
\end{prop}

\begin{remark}
    \label{rem:finite-temp-GCE-TB}
    We may also consider minimising $\mathcal E_\beta$ over 
    \begin{align}
        \mathcal K^{\mathrm{TB}}_{N_\mathrm{el}} &\coloneqq 
        \left\{ 
            \phi^{\mathrm T} \gamma_{\mathrm{TB}} \phi \colon 
            \gamma_\mathrm{TB} \in (\mathbb R^{N_\mathrm{b}\times N_{\mathrm{b}}})^{M\times M}, \,\, 
            \gamma_{\mathrm{TB}}^{\mathrm T} = \gamma_{\mathrm{TB}}, \,\,
            0\leq \gamma_{\mathrm{TB}} \leq 1, \,\,
            \mathrm{Tr}\, \gamma_{\mathrm{TB}} = N_{\mathrm{el}}
        \right\} \nonumber \\
        &\subset \mathcal K_{N_{\mathrm{el}}},
        \label{eq:K-TB}
    \end{align}
   leading to minimisers satisfying $\gamma^0 = \phi^{\mathrm T} F_{\beta,\ep_{\mathrm{F}}}\big( \Ham^{\mathrm{TB}}[\rho_{\gamma^0}]\big) \phi$. 
   
   Moreover, minimising 
   $\mathcal G_{\beta,\mu}$ (as defined in \cref{eq:ab-initio-beta})
   over 
   $\mathcal K^{\mathrm{TB}} 
   \coloneqq \bigcup_{n\geq 0} \mathcal K_n^{\mathrm{TB}}$ 
   leads to the Euler--Lagrange equation $\gamma^0 = \phi^T F_{\beta,\mu}\big( \Ham^{\mathrm{TB}}[\rho_{\gamma^0}]\big)\phi$.
\end{remark}

\subsection{Machine Learning}
\label{sec:ML}
Starting from the \textit{ab initio} energy \cref{eq:ab-initio}, we have considered a general local density approximation for the exchange correlation function. Then, restricting the class of admissible one particle density operators to a local basis of atomic-like orbitals, we obtain a tight binding approximation \cref{eq:TB-energy}. In this way, the total energy of the system is now a function of the (finite dimensional) matrices 
$\gamma_\mathrm{TB} \in (\mathbb R^{N_\mathrm{b}\times N_\mathrm{b}})^{M\times M}$. 
We now restrict the class of admissible density operators to 
$\mathcal P_{N_{\mathrm{el}}}^{\mathrm{ML}} \subset \mathcal P_{N_{\mathrm{el}}}^{\mathrm{TB}}$ 
(or $\mathcal K_{N_{\mathrm{el}}}^{\mathrm{ML}} \subset \mathcal K_{N_{\mathrm{el}}}^{\mathrm{TB}}$ for finite Fermi-temperature models), 
a collection of density operators defined in terms of atom-centred features. 

\begin{remark}[BpopNN]
    \label{rem:BpopNN}
    In the BpopNN \cite{Xie2020}, the authors take the following approach: for Becke populations $p = \{p_1,\dots,p_M\}$, a corresponding one-particle density operator $\gamma(p)$ belongs to the set
    \[
        \argmin\Big\{ \mathcal E(\gamma) \colon
            \gamma \in \mathcal P_{N_\mathrm{el}}, \,\, 
            p_i = \int \rho_\gamma(x) \omega(x - \bm r_i) \mathrm{d}x
        \Big\}
    \]
    where $\omega$ is a smooth weight function concentrated about the origin.
    
    In this way, the density operator $\gamma(p)$ is a function of the atom centred ``features'' $\{p_i\}$, albeit a complicated and likely multi-valued function.
\end{remark}

Motivated by the BpopNN idea but to allow a rigorous analysis we define an explicit mapping from atom-centred features to an admissible density operator. To do this, we will write the Hamiltonian of the system as $H_0 + v$ where $H_0$ is given by \cref{eq:TB-H0} and $v$ is the effective potential written as a function of atom centred features. This approach is similar to that of \cite{HerbstLevitt2022,Gonze1996}.

Given an atom-centred feature vector 
$\widehat{v} = \{\widehat{v}_{mc}\}$ 
where $m=1,\dots,M$ is the atomic index and $c$ belongs to a finite set indexing the features, we write the parameterised effective potential as 
%
$v = \sum_{mc} \widehat{v}_{mc} \int \phi \otimes \phi \, \vb_{mc}$ 
where 
$\vb = \{\vb_{mc} \coloneqq \vb_{c}(\,\cdot - \bm r_m)\}$ 
is a collection of atom-centred functions. We let $I \coloneqq \{(m,c)\}$ be the index set for the atom centred features and write
$\widehat{v}_m \coloneqq \{ \widehat{v}_{mc} \}_{c}$ 
for the feature vector on atom $m$. Throughout this paper, we will assume that these atom-centred functions are exponentially localised: there exists $c_\vb, \eta_\vb > 0$ such that 
\begin{align}
    \label{eq:v-basis-assumption}
    \big|\partial^{\alpha}\vb_{c}(x)\big| \leq c_\vb e^{-\eta_\vb  |x|},
\end{align}
for all $c$ and multi-indices $\alpha \in \mathbb N^3$ with $|\alpha|_1 \leq \nu$ for some $\nu \geq 2$.

Then, as a function of the features, the tight binding Hamiltonian takes the form $H_0 + v$ where $H_0$ is given by \cref{eq:TB-H0}. We now write 
$\lambda_1(v) \leq \lambda_2(v) \leq \dots \leq \lambda_{M\cdot N_\mathrm{b}}(v)$ 
for the ordered eigenvalues of $H_0 + v$ (counting multiplicities), and $C^i(v)$ the corresponding normalised eigenvectors. Under this choice, we define
\begin{align}
    &\mathcal P_{N_\mathrm{el}}^{\mathrm{ML}} \coloneqq
    \left\{ 
        \phi^\mathrm{T}
        \big[ 
            \textstyle\sum_{i=1}^{N_{\mathrm{el}}} C^i(v) \otimes C^i(v)
        \big] 
        \phi \colon \,\, 
        \widehat{v} \in \ell^\infty( I ) 
    \right\} \nonumber\\
    &\subset \mathcal P_{N_\mathrm{el}}^{\mathrm{TB}} 
    = \left\{ 
        \phi^\mathrm{T} \gamma_{\mathrm{TB}} \phi \colon \,\,
        \gamma_\mathrm{TB} \in (\mathbb R^{N_\mathrm{b}\times N_{\mathrm{b}}})^{M\times M}, \,\, 
        \gamma_{\mathrm{TB}}^{\mathrm T} = \gamma_{\mathrm{TB}}, \,\,
        \gamma_{\mathrm{TB}}^2 = \gamma_{\mathrm{TB}},\,\,
        \mathrm{Tr}\, \gamma_{\mathrm{TB}} = N_{\mathrm{el}}
    \right\}\nonumber\\
    &\subset \mathcal P_{N_\mathrm{el}} 
    = \left\{ 
        \gamma \in \mathcal S\big( L^2(\mathbb R^3) \big) \colon \,\,
        \gamma^2 = \gamma, \,\, 
        \mathrm{Tr}\,\gamma = N_{\mathrm{el}}, \,\, 
        \mathrm{Tr} \big[ - \tfrac{1}{2} \Delta \gamma \big] < \infty 
    \right\}.
\end{align}

To simplify the presentation, we will consider \textit{non-degenerate} 
$\gamma \in \mathcal P_{N_\mathrm{el}}^{\mathrm{ML}}$, 
where there is a gap in the spectrum: if 
$\gamma = \phi^{\mathrm T} \big[ \sum_{i=1}^{N_{\mathrm{el}}} C^i(v) \otimes C^i(v) \big] \phi$, 
for some $\widehat{v}\in \ell^\infty(I)$, 
we suppose that 
 $\lambda_{N_{\mathrm{el}}}(v) < \lambda_{N_{\mathrm{el}} + 1}(v)$. In this case, 
%
    %
    %
    %
    %
%
we have
\begin{gather}
    \label{eq:ML-energy}
    \mathcal E(\gamma) = 
    \mathrm{Tr} \,H_0 \chi_{(-\infty,\ep_\mathrm{F})}(H_0 + v)
    + \int \big[
        \rho_\gamma V^\mathrm{nuc} + \tfrac{1}{2} \rho_{\gamma} v_{\rm C} \rho_{\gamma} + \ep_\mathrm{xc}( \rho_{\gamma} ) 
    \big]  \\
    \label{eq:ML-density}
    \text{where} \qquad 
    \rho_\gamma(x) = 2\phi(x)^{\mathrm T} \chi_{(-\infty,\ep_{\mathrm{F}})}(H_0 + v) \phi(x),
\end{gather}
and the Fermi-level $\ep_{\mathrm{F}}$ is any value satisfying $\lambda_{N_{\mathrm{el}}}(v) < \ep_{\mathrm{F}} < \lambda_{N_{\mathrm{el}}+1}(v)$. 
In the following, we will denote by 
$D\chi_{(-\infty,\mu)}$ 
the Jacobian of the map 
$(\mathbb R^{N_\mathrm{b} \times N_\mathrm{b}})^{M\times M} \to (\mathbb R^{N_\mathrm{b} \times N_\mathrm{b}})^{M\times M}$ 
defined by $\Ham \mapsto \chi_{(-\infty,\mu)}(\Ham)$. Moreover, we write 
$V_{\mathrm{eff}}[\rho] \coloneqq V^\mathrm{nuc} + v_{\rm C}\rho + \ep_{\mathrm{xc}}^\prime(\rho)$ 
and
$[\nabla_I v]_{mc} \coloneqq 
    \frac
        {\partial}
        {\partial \widehat{v}_{mc}} 
    \big[ 
        \sum_{m^\prime c^\prime} \widehat{v}_{m^\prime c^\prime} \int \phi \otimes \phi \, \vb_{m^\prime c^\prime} 
\big] 
= \int \phi \otimes \phi \, \vb_{mc}$.

The minimisation of \cref{eq:ML-energy} over $\mathcal P_{N_\mathrm{el}}^{\mathrm{ML}}$ is well-posed and yields an Euler--Lagrange equation: 
\begin{lemma}
    \label{ML-minimiser}
    There exists  
    $\gamma^0 = \phi^{\mathrm T}
    \big[ 
        \sum_{n=1}^{N_{\mathrm{el}}}C^n(v) \otimes C^n(v)
    \big] \phi$,  
    a minimiser of $\mathcal E$ over $\mathcal P_{N_\mathrm{el}}^{\mathrm{ML}}$. Moreover, if  $\lambda_{N_{\mathrm{el}}}(v) < \ep_{\rm F} < \lambda_{N_{\mathrm{el}}+1}(v)$, then we have 
    \begin{align}
        \label{eq:ML-EL}
        \Braket{ - D\chi_{(-\infty,\ep_\mathrm{F})}\big( H_0 + v \big) 
        \Big[ 
            v - \int \phi \otimes \phi \, V_{\mathrm{eff}}[\rho_{\gamma^0}]  
        \Big], \nabla_{I} v } = 0.
    \end{align}
    %

    In that case, we also have the following approximation result
    \begin{align}
        \label{eq:ML-approx}
        \left|
            v - \int \phi \otimes \phi\, V_{\mathrm{eff}}[\rho_{\gamma^0}]
        \right|
        \lesssim \min_{V \in \mathrm{span} \, \vb}
        \left|
            \int \phi \otimes \phi \Big[ V - V_{\mathrm{eff}}[\rho_{\gamma^0}] \Big]
        \right|.
    \end{align}
    Therefore, if $\vb$ is a complete basis, we have
    \begin{align}
        \label{eq:ML-complete}
        \gamma^0 
        = \phi^{\mathrm T} 
        \chi_{(-\infty, \ep_\mathrm{F})}(\Ham^\mathrm{TB}[\rho_{\gamma^0}]) 
        \phi
    \end{align}
    where $\Ham^\mathrm{TB}$ is given by \cref{eq:Ham[rho]}. 
\end{lemma}

\begin{proof}[Sketch of the proof] 
    First, we note that $\mathcal E$ is continuous and the space of possible 
    $\{C^i(v)\}_{i=1}^{N_{\mathrm{el}}}$ for $\widehat{v}\in \ell^\infty(I)$ 
    is compact. In particular, there exists a minimiser. Supposing there is a gap in the spectrum, we obtain \cref{eq:ML-EL} as an Euler--Lagrange equation. We finally obtain \cref{eq:ML-approx,eq:ML-complete} by noting that 
    $D\chi_{(-\infty, \ep_{\mathrm{F}})}(H_0 + v)$ 
    is invertible.  
\end{proof}

\begin{remark}[Finite Fermi-temperature]
    \label{rem:lemma-finite-temp}
    In the finite Fermi-temperature case, we minimise $\mathcal E_\beta$ as defined in \cref{eq:ab-initio-beta} over the space of density operators of the form $\gamma = \phi^T F_{\beta,\ep_{\mathrm{F}}}(H_0 + v) \phi$ where $\widehat{v}\in\ell^\infty(I)$ and $\ep_{\mathrm{F}}$ is the unique solution to $\mathrm{Tr} \, F_{\beta, \ep_\mathrm{F}}(H_0 + v) = N_{\mathrm{el}}$. Minimisers $\gamma^0$ satisfy the following Euler--Lagrange equation:
    \begin{align}
        \label{eq:ML-EL-beta}
        \Braket{ - DF_{\beta, \ep_\mathrm{F}}(H_0 + v) 
        \Big[ 
            v - \int \phi \otimes \phi \, V_{\mathrm{eff}}[\rho_{\gamma^0}]  
        \Big], \nabla_{I} v } = 0.
    \end{align}
    The full proof of this result is contained in the proof of \cref{ML-minimiser} in \S\ref{sec:proofs}.
\end{remark}

\begin{remark}[Grand canonical ensemble]
    \label{rem:GCE}
    For the grand canonical model, we minimise $\mathcal G_{\beta,\mu}$ as defined in \cref{eq:ab-initio-GCE} over the space of density operators of the form $\gamma = \phi^T F_{\beta,\mu}(H_0 + v) \phi$ where $\widehat{v}\in\ell^\infty(I)$. Minimisers satisfy the Euler--Lagrange equation \cref{eq:ML-EL-beta} but with $\ep_{\mathrm{F}} = \mu$.  
    The proof of this result follows the exact same ideas as in the proof of \cref{ML-minimiser}.  
\end{remark}

\section{Results}
\label{sec:results}
In the foregoing section, we reviewed a general framework for parameterizing electronic structure models. In the present section we show how the ``global many-body'' functionals arising in those models can be decomposed into low-dimensional and spatially localized components. Although our results can be formally applied in the general setting of the previous section, for the sake of simplicity of presentation, we give rigorous statements only for the grand-canonical ensemble setting. We will comment on how to incorporate the charge constraint into the analysis in remarks following the main results. 

We fix the chemical potential $\mu$ and inverse Fermi-temperature $\beta \in (0,\infty]$ and consider the total free energy~\cref{eq:ab-initio-GCE} but now written explicitly as a function of the atomic environment and additional atomic features $\widehat{v}$:
\begin{align}
   \mathcal G(\bm r, \widehat{v}) 
    &\coloneqq \mathcal G_{\beta,\mu}\big( \phi^T F_{\beta,\mu}(H_0 + v) \phi \big) \nonumber \\
    &= \mathrm{Tr} \Big[ 
        (H_0-\mu) F_{\beta,\mu}(H_0 + v)
        + \beta^{-1} S\big( F_{\beta,\mu}(H_0 + v) \big)
    \Big] \nonumber \\
    &\qquad + \int \big[
        \rho(\bm r, \widehat{v}) V^\mathrm{nuc} + \tfrac{1}{2} \rho(\bm r, \widehat{v}) v_{\rm C} \rho(\bm r, \widehat{v}) + \ep_\mathrm{xc}( \rho(\bm r, \widehat{v}) ) 
    \big]
    \label{eq:ML-energy-GCE}
\end{align}
where $\rho(\bm r, \widehat{v};x) \coloneqq 2\phi(x)^T F_{\beta,\mu}(H_0 + v) \phi(x)$. 

To simplify the notation, we will from now on omit the parameters $(\beta,\mu)$ in $\mathcal G(\bm r,\widehat{v})$, and will similarly also write $F(x) \coloneqq F_{\beta,\mu}(x)$. Moreover, when $\widehat{v} \in \ell^\infty(I)$ is fixed and it is clear from the context, we will write $\Ham = H_0 + v$,
\[
    \mathsf{g}_- \coloneqq \dist\big( \mu, \sigma(\Ham) \cap (-\infty, \mu] \big),
    \quad \textrm{and} \quad 
    \mathsf{g}_+ \coloneqq \dist\big( \mu, \sigma(\Ham) \cap [\mu,\infty) \big).
\]
The spectral gap will be written $\mathsf{g} \coloneqq \mathsf{g}_- + \mathsf{g}_+ \geq 0$.

\subsection{Locality}
We first show that the total energy may be written as a sum of local energy contributions together with an explicit long-range electrostatic term:
\begin{theorem}
    \label{locality}
    Fix $\bm r$ and $\widehat{v}\in \ell^\infty(I)$ and suppose that either $\beta < \infty$ or that $\mathsf{g}_{\pm} > 0$. 
    Then, the total energy \cref{eq:ML-energy-GCE} may be decomposed as follows:
    \begin{align}
        \label{eq:locality-decomp}
        \mathcal G(\bm r, \widehat{v}) &= 
        \sum_\ell E_\ell(\bm r, \widehat{v}) +
        E_{\mathrm{el}}\big[ \rho(\bm r, \widehat{v}) \big] 
        \qquad \text{and} \qquad
        \rho(\bm r, \widehat{v}) = \sum_\ell \rho_\ell(\bm r, \widehat{v}), 
    \end{align}
    where 
    $E_{\mathrm{el}}[\rho] 
    \coloneqq \int \big[
        \rho V^\mathrm{nuc} + \tfrac{1}{2} \rho v_{\rm C} \rho
    \big]$ 
    and there exists $\eta >0$ such that
    \begin{align}
        \label{eq:locality}
        \begin{split}
        \left|
            \frac{\partial E_\ell(\bm r, \widehat{v})}{\partial \bm r_k}
        \right|
        + 
        \left|
            \frac{\partial E_\ell(\bm r, \widehat{v})}{\partial \widehat{v}_k}
        \right|
        &\lesssim e^{-\eta \,r_{\ell k}} \\
        \bigg|
            \frac{\partial \rho_\ell(\bm r, \widehat{v};x)}{\partial \bm r_k}
        \bigg|
        + 
        \bigg|
            \frac{\partial \rho_\ell(\bm r, \widehat{v};x)}{\partial \widehat{v}_k}
        \bigg|
        &\lesssim e^{-\eta\, [ |x - \bm r_\ell| + r_{\ell k} + |\bm r_k - x|]}.
        \end{split}
    \end{align}
    The exponent satisfies $\eta \sim \beta^{-1} + \min\{\mathsf{g}_-,\mathsf{g}_+\}$ 
    as $\beta^{-1} + \mathsf{g}\to 0$. 
    If $\beta = \infty$, we have $\eta \sim \mathsf{g}$ as $\mathsf{g}\to 0$.

    Similar estimates hold for higher derivatives as long as $\nu$ as in \cref{eq:phi-assumption,eq:v-basis-assumption} is sufficiently large. 
\end{theorem}
\begin{proof}[Sketch of the proof]
    First, we take a disjoint union $\dot{\bigcup}{}_{\ell = 1}^{M} \mathcal N_\ell = \mathbb R^3$, and define
    \begin{align}
         E_\ell(\bm r, \widehat{v}) 
         &= \mathrm{tr} \Big[ 
            (H_0-\mu) F(H_0 + v)
            + \beta^{-1} S\big( F(H_0 + v) \big)
        \Big]_{\ell\ell} + \int_{\mathcal N_\ell} \ep_\mathrm{xc}( \rho(\bm r, \widehat{v}) ) \label{eq:Eell-sketch-1}\\
        &= \mathrm{tr} \Big[ 
            G(H_0 + v)
            - v F(H_0 + v)
        \Big]_{\ell\ell} + \int_{\mathcal N_\ell} \ep_\mathrm{xc}( \rho(\bm r, \widehat{v}) )
        \label{eq:Eell-sketch-2}
    \end{align}
    where $G(x) \coloneqq \beta^{-1} \log\big( 1 - F(x) \big)$. For $\beta < \infty$, \cref{eq:Eell-sketch-2} is more convenient, while at $\beta = \infty$, we consider \cref{eq:Eell-sketch-1} (after disregarding the entropy).
    Moreover, we define 
    \begin{align}
        \label{eq:rhoell-sketch}
        \rho_\ell(\bm r, \widehat{v};x) \coloneqq 2\sum_k \phi_\ell(x)^{\mathrm{T}} F(H_0 + v)_{\ell k} \phi_k(x).
    \end{align} 
    To conclude, we can apply the main ideas of previous works \cite{ChenOrtnerThomas2019:locality,Thomas2020:scTB,ChenOrtner16,ChenLuOrtner18} to bound the derivatives of $\Ham \mapsto F(\Ham)$ and $\Ham\mapsto G(\Ham)$ evaluated at $\Ham = H_0 + v$. A full proof is given in \cref{sec:proofs}.
\end{proof}

\begin{remark}[Canonical ensemble]
    The locality estimates break down in the canonical ensemble at finite temperature because the choice of Fermi-level introduces non-locality into the system~\cite{ChenLuOrtner18}.
\end{remark}

\subsection{Body-ordered approximation}

Now, approximating $F$ and $G$ as in \cref{eq:Eell-sketch-1} (for $\beta = \infty$), \cref{eq:Eell-sketch-2} (for $\beta < \infty$), and \cref{eq:rhoell-sketch} by polynomials $F_N$ and $G_N$ of degree at most $N$, we obtain the following approximate energy: 
\begin{align}
    \label{eq:decomp-N}
    \mathcal G_N(\bm r, \widehat{v}) &= \sum_\ell E_{N,\ell}(\bm r, \widehat{v}) + E_{\mathrm{el}}\big[ \rho_N(\bm r, \widehat{v}) \big] 
    \qquad \text{and} \qquad
    \rho_N(\bm r, \widehat{v}) = \sum_\ell \rho_{N,\ell}(\bm r, \widehat{v}) 
\end{align}
where 
\begin{align}
    \rho_{N,\ell}(\bm r, \widehat{v};x) 
    &\coloneqq2 \sum_k \phi_\ell(x)^{\mathrm T} F_{N}\big( H_0 + v \big)_{\ell k} \phi_k(x), \qquad \text{and} \nonumber\\
    E_{N,\ell}(\bm r, \widehat{v}) 
     &= \begin{cases} \mathrm{tr} \Big[ 
        (H_0-\mu) F_N(H_0 + v)
    \Big]_{\ell\ell} + \int_{\mathcal N_\ell} \ep_\mathrm{xc}( \rho_N(\bm r, \widehat{v}) ), &\text{if } \beta = \infty\\
    \mathrm{tr} \Big[ 
            G_N(H_0 + v)
        - v F_N(H_0 + v)
    \Big]_{\ell\ell} + \int_{\mathcal N_\ell} \ep_\mathrm{xc}( \rho_N(\bm r, \widehat{v}) ), &\text{if } \beta < \infty.
    \end{cases} \nonumber 
\end{align}

\begin{remark}(Locality of the body--ordered approximations) \label{rem:locality-N}
    Following the same proof as Theorem~\ref{locality}, one may show that the body-ordered approximations also satisfy the locality estimates \cref{eq:locality}.
\end{remark}



The decomposition \cref{eq:decomp-N} is a nonlinear body-order expansion: 
    \begin{theorem}
        \label{body-order}
        $E_{N,\ell}(\bm r, \widehat{v})$ and $\rho_{N,\ell}(\bm r,\widehat{v})$ are functions of body-order at most $2N + 1$. 
        
        Specifically, with $\bm u_{\ell} \coloneqq ( \bm r_\ell, Z_\ell, \widehat{v}_{\ell} )$
        there exist 
        $U_{nN}^{(l)}(\bm u_{k_1}, \dots, \bm u_{k_n}) \in (\mathbb R^{N_\mathrm{b} \times N_\mathrm{b}})^{M\times M}$ 
        for $n = 0,\dots, 2N-1$ and $l = 1,2$ 
        and a nonlinear function $\ep_\ell \colon (\mathbb R^{N_\mathrm{b} \times N_\mathrm{b}})^{M\times M} \to \mathbb R$
        such that, on defining
    \begin{gather}
        \bm U_N^{(l)}
        \coloneqq U_{0N}^{(l)} 
        + \sum_{k} U_{1N}^{(l)}(\bm u_k) 
        + \dots + 
        \sum_{
            \above
            {k_1, \dots, k_{2N-1}} 
            {k_i \not= k_j \, \forall i\not=j}
        } 
        U_{2N-1, N}^{(l)}(\bm u_{k_1},\dots, \bm u_{k_{2N-1}})
        \label{eq:U}
    \end{gather}
    we obtain
    \begin{align}
        \rho_{N,\ell}(\bm r, \widehat{v}; x) 
        &= \phi_\ell(x)^{\mathrm T} 
        \big[
            \bm U_N^{(1)} \phi(x) 
        \big]_{\ell\ell} 
        \qquad \text{and} \label{eq:rho-phi-U-phi}\\ 
        E_{N,\ell}(\bm r, \widehat{v})
        &= \tr \big[
            \bm U_{N}^{(2)}
        \big]_{\ell\ell} 
        + \ep_\ell\big( \bm U_{N}^{(1)} \big). 
        \label{eq:E-tr+ep}
    \end{align} 
    Moreover, one can write $U_{nN}^{(l)}(\bm u_{k_1},\dots, \bm u_{k_{n}})_{\ell k} 
    = \widetilde{U}_{nN}^{(l)}(\bm u_{\ell k}; \bm u_{\ell  k_1},\dots, \bm u_{\ell k_{n}})$ 
    where 
    $\bm u_{\ell k} \coloneqq (\bm r_{\ell k}, Z_{\ell}, Z_k, \widehat{v}_\ell, \widehat{v}_k)$.
    \end{theorem}
    \begin{proof}[Sketch of the proof]
        Since entries of the Hamiltonian $H_0 + v$ have body-order at most $3$, for all polynomials $P_N$ of degree at most $N$, we have that $P_N(H_0 + v)_{\ell k}$ are quantities of body-order at most $2N+1$. In particular, the quantities 
        \[
            \rho_{N,\ell}(\bm r, \widehat{v}),
            \quad 
            \mathrm{tr} \big[ (H_0 - \mu) F_N(H_0 + v) \big]_{\ell\ell},
            \quad \textrm{and} \quad 
            \mathrm{tr} \big[ G_N(H_0 + v) - v F_N(H_0 + v)\big]_{\ell\ell}
        \]
        have body-order at most $2N+1$.
    \end{proof}

\begin{remark}[Canonical ensemble]
    \label{rem:CE}
    In the canonical ensemble, a possible choice for the body-ordered approximation to the Fermi-level, $\ep_{\mathrm{F},N}$, is a solution to the equation 
    \[ 
        \mathrm{Tr}\,F_{N,\beta,\ep_{\mathrm{F},N}}\big(H_0 + v\big) = N_{\mathrm{el}}
    \]
    where $F_{N,\beta,\mu}$ is a degree $N$ polynomial approximation to $F_{\beta,\mu}$. Therefore, one must undergo an additional pre-computation step to find an approximate Fermi-level. In this way, the approximate Fermi-level is a nonlinear body-ordered function. That model can still be decomposed into parameterized components that are local and body-ordered components, but the operations performed on those components become more complex.
\end{remark}

\subsection{Convergence of minimisers}
We first note that critical points of the approximate energy satisfy the following Euler--Lagrange equation:
\begin{lemma}
    \label{ML-minimiser-N}
    If $\beta < \infty$, then we approximate $G$ with a polynomial $G_N$ of degree at most $N$, and define $F_N \coloneqq G_N^\prime$. On the other hand, if $\beta = \infty$, then we simply approximate $F$ with a polynomial $F_N$ of degree at most $N$.
    
    Suppose $\widehat{v}^\star$ minimises $\mathcal G_N(\bm r, \widehat{v})$ (at the fixed geometry $\bm r$). Then, 
    \begin{align}
        \label{eq:ML-EL-N}
        \Braket{ - DF_N\big( H_0 + v^\star \big) 
        \Big[ 
            v^\star - \int \phi \otimes \phi \, V_{\mathrm{eff}}[\rho_{N}(\bm r, \widehat{v}^\star)]  
        \Big], \nabla_{I} v^\star } = 0.
    \end{align}
\end{lemma}
\begin{proof}[Sketch of the proof]
    The proof is analogous to that of \cref{ML-minimiser}. In~\cref{sec:proofs} we briefly outline the minor differences.
\end{proof}

We now consider stable critical points $\widehat{v}^\star$ of $\mathcal G$: 
\begin{align}\label{eq:stable-critical}
    \nabla \mathcal G(\widehat{v}^\star) = 0, 
    \qquad \text{and} \qquad
    \nabla^2 \mathcal G(\widehat{v}^\star) \text{ invertible}.
\end{align}

We wish to approximate $\widehat{v}^\star$ and (hence the corresponding electron density $\rho(\bm r, \widehat{v}^\star)$) with solutions $\widehat{v}^\star_N$ to the approximate energy $\mathcal G_N$ (and the corresponding density $\rho_N(\bm r, \widehat{v}^\star_N)$). To do so, we require the following result about the convergence of the polynomial approximations $F_N$ and $G_N$:
\begin{lemma}
    \label{poly-approx}
    Fix $\mu \in \mathbb R$, $\beta \in(0,\infty]$, 
    $X \coloneqq [\underline{\sigma}, \mu - \mathsf{g}_-] \cup [\mu + \mathsf{g}_+, \overline{\sigma}] \subset \mathbb R$, 
    and $\mathsf{g} \coloneqq \mathsf{g}_- + \mathsf{g}_+$.
    Then, for $N$ sufficiently large, there exists $\theta > 0$ and polynomial approximations $G_N$ and $F_N \coloneqq G_N^\prime$ such that
    \[
        \big\|  G - G_N \big\|_{L^\infty(X)}
        + \big\| F - F_N \big\|_{L^\infty(X)}
        \lesssim e^{-\theta N},
    \]
    where $\theta \sim \beta^{-1} + \sqrt{\mathsf{g}_-}\sqrt{\mathsf{g}_+}$ as $\beta^{-1} + \mathsf{g}\to 0$. If $\beta = \infty$, we have $\theta \sim \mathsf{g}$ as $\mathsf{g}\to 0$. 
\end{lemma}
\begin{proof}
    The proof uses logarithmic potential theory, see \cite{Saff2010} for an overview of the general setting. For details of this theory applied to the particular functions $F$ and $G$, see \cite{ThomasChenOrtner2022:body-order}.
\end{proof}

With this result at hand, we define the $\mathcal G_N$ as in \cref{eq:decomp-N} using $F_N$ and $G_N$ from Lemma~\ref{poly-approx}. Then, we may approximate $\widehat{v}^\star$ with critical points of $\mathcal G_N$. The following result is a direct consequence of regular perturbation theory.

\begin{theorem}
    \label{minimisers}
    Suppose $\widehat{v}^\star$ is a stable critical point of $\mathcal G$. Then, there exists $\overline{C} > 0$ such that for $\theta N > 2 \log M + \overline{C}$, there exists a critical point $\widehat{v}_N^\star$ of $\mathcal G_{N}$ such that
    \[
        \| \widehat{v}_N^\star - \widehat{v}^\star \|_{\ell^\infty}
        +
        \big\|\rho_N(\bm r, \widehat{v}^\star_N) - \rho(\bm r, \widehat{v}^\star)\big\|_{L^\infty(\mathbb R^3)}
        \lesssim M e^{-\theta N}.
    \]
    Moreover, we have 
    \[
        \big| \mathcal G_N(\bm r, \widehat{v}_N^\star) - \mathcal G(\bm r, \widehat{v}^\star) \big|
        \lesssim M^2 e^{-\theta N}.
    \]
\end{theorem}
\begin{remark}
    The dependence on system size $M$ is a weakness of our current results, preventing the analysis from being extended to the bulk limit. It is unclear to us whether this is a worst-case that is in principle attainable for some systems, or whether a stronger generally valid estimate can be proven. This is a possible direction for future study.
\end{remark}

\section{Numerical Experiments}
\label{sec:numerics}
We now present a preliminary numerical study to demonstrate that one can construct systematically improvable machine learning models for systems with long-range charge equilibration effects, based on the results of \S\ref{sec:results}.   
\begin{enumerate}
    \item There exist functions $E_{N,\ell}(\bm{r}, \widehat{v})$ and $\rho_{N,\ell}(\bm{r}, \widehat{v})$ which are local in the sense of Theorem~\ref{locality} and can be used to approximate the total free energy as in \cref{eq:locality-decomp},
    
    \item By Theorem~\ref{body-order}, $E_{N,\ell}(\bm{r}, \widehat{v})$ and $\rho_{N,\ell}(\bm{r}, \widehat{v})$ have body-order of at most $2N+1$ in the features $\bm{u}_k = (\bm{r_k}, Z_k, \widehat{v_k})$,
    
    \item Finally, the approximate free energy converge with increasing body-order, as stated rigorously in Theorem~\ref{minimisers}.
\end{enumerate}
To demonstrate these statements in practice, one can extract the total energy, electron density and atomic features $\bm{u}_k$ from quantum mechanical calculations and fit the functions $E_{N,\ell}(\bm{r}, \widehat{v})$ and $\rho_{N,\ell}(\bm{r}, \widehat{v})$ using machine learned, body-ordered functions.

\subsubsection*{Implementation details}

First, we make the following design choices:
\begin{itemize}
\item The local basis functions $\phi_{\ell a}$, in which the electron density is expanded, are Gaussian-type orbitals and are the same for all chemical elements: 
\begin{align*}
    \phi_{\ell a} \coloneqq \phi_{a}(\,\cdot - \bm r_\ell)
    \quad \text{where} \quad 
    \phi_{a}(x) \coloneqq 
    \phi_{\lambda\mu}(x) \coloneqq \frac{1}{\mathcal{N}} e^{- \frac{x^2}{2\sigma^2}}R_{\lambda\mu}(x)
\end{align*}
where $(\lambda,\mu)$ is an angular momentum index tuple, $R_{\lambda\mu}$ is a real solid harmonic.

\item The effective potential features $\widehat{v}_{\ell}$ and the effective potential $V(x)$ are connected via a projection:
\begin{align} 
\widehat{v}_{\ell} = \{\widehat{v}_{\ell,n\lambda\mu}\}_{n\lambda\mu}, \quad \widehat{v}_{\ell,n\lambda\mu} = \int V(x) \psi_{n\lambda\mu}(x-\bm{r}_\ell)\mathrm{d}x
\label{eq:vhat-projection}
\end{align} 
where $\psi_{n\lambda\mu}$ is once again a Gaussian type orbital, but now with variable width $\sigma_n$ indexed by $n$. As the range of $n$ and $\lambda$ grows, more information about the local effective potential is encoded in $\widehat{v}_{\ell}$.

\begin{remark}
    For the analysis, we needed $V$ to be a function of $\widehat{v}$ so found it more convenient to work with the expansion, $V(x) = \sum_{\ell a} \widehat{v}_{\ell a} \Xi_{\ell a}(x)$ rather than the projection \cref{eq:vhat-projection}. In particular, the problem of solving \cref{eq:vhat-projection} for $V$ is a complicated ill-posed inverse problem. However, both these approaches are equivalent if both sets of basis functions $\{\Xi_{\ell a}\}$ and $\{\phi_{\ell a}\}$ are orthogonal.
\end{remark}

\item The effective potential $V = V_{\mathrm{eff}}[\rho]$ is evaluated using a smeared nuclei approximation, and by also neglecting the exchange correlation contribution: 
\begin{align}
    V_{\mathrm{eff}}[\rho](x) = \int \frac{\rho(y) - \nu(y)}{|x-y|}
    \mathrm{d}y.
    \label{eq:smeared_nuclei_v_eff}
\end{align}
\end{itemize}

\subsubsection*{Extracting $\rho$ and $\widehat{v}$ from density functional theory (DFT) calculations}

A coarse-grained electron density was extracted by partitioning the total electron density of the DFT calculation, $\rho^{\rm{DFT}}$, 
onto atom centered functions with angular momentum indices ($\lambda,\mu$), giving atomic charges, dipoles and quadrupoles. Further details on this process are given in \Cref{sec:appendix-fhi-aims}. Terms with $\lambda>2$ were not used. 

These atom centered quantities are then interpreted as the coefficients $\rho_{\ell, \lambda \mu}^{\mathrm{DFT}}$ in a coarse-grained model of the system:
\begin{align}
    \rho^{\rm{DFT}} = \sum_{\ell,\lambda\mu} \rho_{\ell,\lambda\mu}^{\rm{DFT}} \phi_{\lambda\mu}(\,\cdot - \bm r_\ell). \label{eq:DFT-rho}
\end{align}
The task of predicting the charge density is then to predict the coefficients $\rho_{\ell, \lambda \mu}^{\mathrm{DFT}}$. From this surrogate charge distribution one can compute a coarse-grained potential at the solution of the DFT calculation through \eqref{eq:smeared_nuclei_v_eff} and thus the features $\widehat{v}_{\ell,\lambda\mu}^{\rm{DFT}}$ via \eqref{eq:vhat-projection}.

\subsubsection*{Parameterisation of $E_\ell$ and $\rho_\ell$} 

$E_{N,\ell}$ and $\rho_{N,\ell}$ are explicitly constructed as body-ordered functions of the features $\bm{u}_k = (\bm{r_k}, z_k, \widehat{v}_k)$ of neighbouring atoms, through a flexible functional form which is fitted to a dataset. In particular, we used elements from the MACE equivariant message passing network architecture~\cite{MACE2022}. Full details are given in Appendix~\ref{sec:appendix-ml-param}. 

\subsubsection*{Results Clusters of Water}

DFT calculations were performed on a dataset of $4800$ small clusters of water molecules with the FHI-aims DFT code \cite{Blum2009}, using the PBE exchange correlation functional \cite{Perdew1996}. The total energy and electron density can be fitted with the above framework. 

The function $\rho_{N,\ell}(\bm{r}, v)$ can be trained directly against $\{\bm{u}_k\}_k$ extracted from DFT. The function $E_\ell(\bm{r}, v)$ is instead fitted via the total energy of the water cluster and its gradients with respect to nuclear coordinates $\{\partial E_{\mathrm{tot}} / \partial \bm{r}_\ell\}_{\ell}$. These gradients are accessible from the DFT calculation as nuclear forces. As in the main result \eqref{eq:locality-decomp}, the total energy also contains the Hartree energy, which can be computed from the coarse-grained density \eqref{eq:DFT-rho}.

Figure \ref{fig:water_convergence_plot} shows the convergence of this scheme with respect to the number of electric potential features per atom (the range of $n$ and $\lambda$ in \eqref{eq:vhat-projection}), and the body-order of the functions $E_{N, \ell}$ and $\rho_{N, \ell}$.

\begin{figure}
    \centering
    \includegraphics[width=\linewidth]{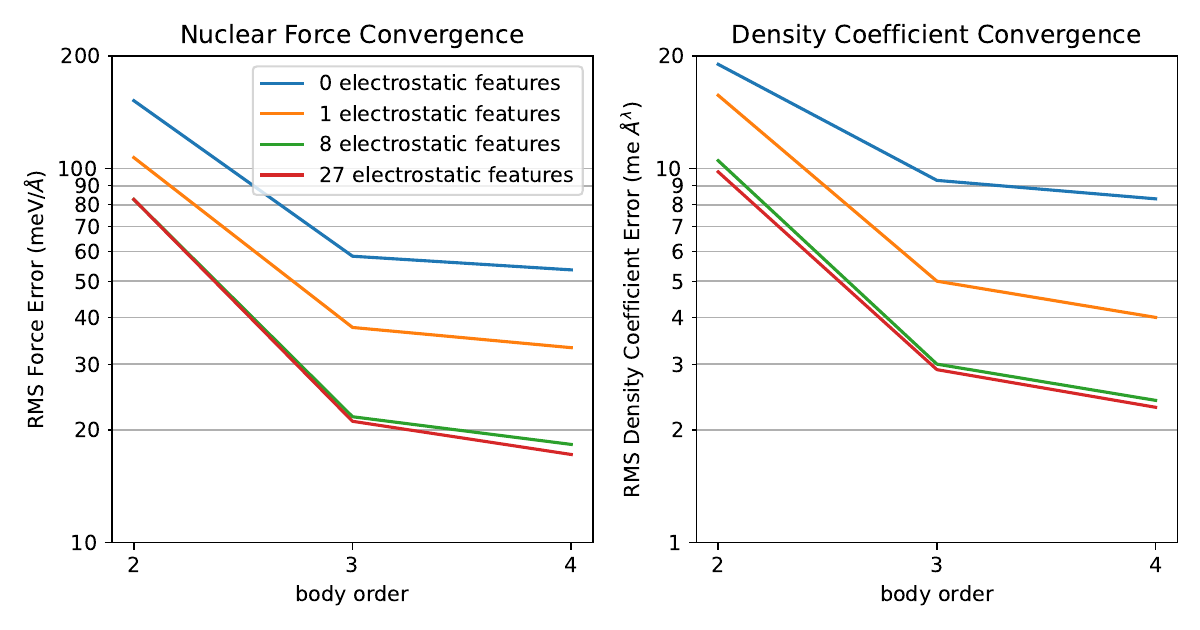}
    \caption{Convergence of the fitted functions $E_{N,\ell}$ and $\rho_{N,\ell}$ as a function of the number of electric potential features $\hat{v}$ per atom and the body-order of the functions. Convergence of the PES is assessed by the accuracy of the nuclear forces. Convergence of the density is assessed by the error in the partitioned density coefficients in equation \eqref{eq:DFT-rho}. These quantities are multipole moments and have units of charge $\times$ length$^\lambda$. The number of electric potential features per atom is restricted by considering only $\{\psi_{n\lambda\mu}\}$ for $n \leq n_{\mathrm{max}}$, $\lambda \leq \lambda_{\mathrm{max}}$, and $-\lambda \leq \mu \leq \lambda$ (so that the number of features is $n_{\mathrm{max}}( \lambda_{\mathrm{max}} + 1)^2)$. For the plots, we take $(n_{\mathrm{max}}, \lambda_{\mathrm{max}}) = (0,0)$, $(1,0)$, $(2,1)$, and $(3,2)$.
    }
    \label{fig:water_convergence_plot}
\end{figure}

\section{Conclusions}
We rigorously demonstrated the possibility of decomposing a general electronic structure potential energy surface into local and body-ordered components, even in the presence of long-range charge equilibration effects. This was achieved by first discretising Kohn--Sham DFT in a basis of localised orbitals, resulting in a tight-binding-like formulation. The key idea then was to introduce an internal atom-centered variable (or, feature), resulting in an extended potential energy landscape representing the effective potential. By equilibrating the new internal variable (i.e., achieving self-consistency) our local model is capable of representing non-trivial long-range effects. Our formulation is closely related to extended variables in classical charge equilibriation \cite{Rappe1991-PhysChem95-QEq,Rick1994}, or the Becke populations in \cite{Xie2020}. We construct body-ordered approximations to the self-consistent solutions and prove an exponential rate of convergence as we increase the maximal body-order.

\section*{Acknowledgements} 
JT is supported by 
EPSRC Grant EP/W522594/1. 
%
CO is supported by NSERC Discovery Grant GR019381 and NFRF Exploration Grant GR022937.
WB thanks the AFRL for funding through grant FA8655-21-1-7010. This work utilized computational resources from the ARCHER2 UK National Supercomputing Service (\url{http://www.archer2.ac.uk}) which is funded by EPSRC the membership of the UK Car-Parrinello Consortium. We also utilised the Cambridge Service for Data Driven Discovery (CSD3).

\section{Proofs}
\label{sec:proofs}

We recall here that $f \lesssim g$ means that $f \leq C g$ for some constant $C$, independent of system size $M$. When we wish to denote explicit dependencies, we use the notation $f \lesssim_{D} g$; there exists a constant $C_D$ depending on $D$ (but not on $M$) such that $f \leq C_D g$.

\subsection{Preliminaries}
\label{sec:proofs-prelim}

We first state an elementary result that we will use throughout: 
\begin{lemma}
    \label{lem:elementary}
    For $0 < \eta_1 \leq \dots \leq \eta_n$ and $\{\bm r_{m}\}_{m=1}^M \subset \mathbb R^d$ where $M\geq n$ or $M= \infty$, we have
    \begin{align}
        \int e^{-\sum_{l=1}^n \eta_l |x - \bm r_{l}|} \mathrm{d}x
        &\lesssim_d (\eta_n)^{-d}
        e^{-\frac{1}{4} \eta_1 [r_{12}+r_{23}+\dots + r_{n-1,n}+ r_{n1}]}, \qquad \text{and} \label{eq:expint}\\
        \sum_{m=1}^M e^{-[\eta_1 r_{1 m} + \eta_2 r_{m2}]}
        &\lesssim_d 
        \left[
            1 + \big( \eta_2 \min\limits_{i\not=j} r_{ij}\big)^{-d}
        \right] 
        e^{-\frac{1}{2} \eta_{1} r_{1 2}} \label{eq:expsum}
    \end{align}
    where $r_{ij} \coloneqq |\bm r_{j} - \bm r_i|$.
\end{lemma}
\begin{proof}
    We first show \cref{eq:expint}. We first note that 
    $\sum_{l=1}^n |x - \bm r_l| 
    \geq \frac{1}{2} \big( r_{12} + r_{23} + \dots + r_{n1} \big)$. 
    In particular, we have 
    \begin{align}
        \int e^{-\sum_{l=1}^n \eta_l |x - \bm r_{l}|} \mathrm{d}x
        &\leq \int e^{-\frac{1}{2}\eta_n |x-\bm r_n|} \mathrm{d}x \,\, 
        e^{-\frac{1}{4} \eta_1 [r_{12}+r_{23}+\dots + r_{n-1,n}+ r_{n1}]} \nonumber\\
        &= 2^d (d-1)! \, \big|\partial B_1\big| (\eta_n)^{-d}
        e^{-\frac{1}{4} \eta_1 [r_{12}+r_{23}+\dots + r_{n-1,n}+ r_{n1}]} \nonumber.
    \end{align} 
    
    We sketch the proof of \cref{eq:expsum}; more details can be found in \oldcite[Appendix~B.3]{PhD}. After defining $\mathfrak{m}\coloneqq \min_{i\not=j} r_{ij}$, we have $\mathcal B_m \coloneqq B_{r_{2m}}(\bm r_2) \cap B_{\mathfrak{m}/2}(\bm r_m)$ are disjoint sets with $r_{2m} \geq |x - \bm r_2|$ for all $x \in \mathcal B_m$ and $|\mathcal B_m| \geq c_d \mathfrak{m}^d$. In particular, we obtain
    \begin{align}
        \sum_{m=1}^M e^{-[\eta_1 r_{1 m} + \eta_2 r_{m2}]}
        &\leq \Big( 
            \sum_{m=1}^M e^{-\frac{1}{2}\eta_2 r_{m2 }}
        \Big) e^{-\frac{1}{2}\eta_1 r_{12}} 
        \leq \Big( 
            1 + \sum_{m\not=2} \frac{1}{\big| \mathcal B_m \big|} \int_{\mathcal B_m} e^{-\frac{1}{2}\eta_2 |x - \bm r_{m}|}
        \Big) e^{-\frac{1}{2}\eta_1 r_{12}} \nonumber\\
        &\leq \Big( 
            1 + \frac{1}{c_d \mathfrak{m}^d} \int e^{-\frac{1}{2}\eta_2 |x|}
        \Big) e^{-\frac{1}{2}\eta_1 r_{12}}
        \leq \Big( 
            1 + \frac{2^d(d-1)! \, \big|\partial B_1\big| }{c_d (\mathfrak{m}\eta_2)^d} 
        \Big) e^{-\frac{1}{2}\eta_1 r_{12}}, \nonumber
    \end{align}
    which concludes the proof.
\end{proof}

For the remainder, we will assume that $\min_{i\not=j} r_{ij} \geq \mathfrak{m} >0$, which allows us to apply \Cref{lem:elementary} with prefactors that only depend on the configuration via $\mathfrak{m}$.

Next, we show that the tight binding Hamiltonian is exponentially localised:
\begin{lemma}
    \label{lem:Hamiltonian}
    Fix $\bm r$ and $\widehat{v}\in \ell^\infty(I)$, and write $\Ham = H_0 + v$. Then, there exists $\eta_\Ham > 0$, such that 
    \begin{align}
        \left| 
            \Ham_{\ell k, ab}
        \right| 
        \lesssim (1 + \|\widehat{v}\|_{\ell^\infty}) e^{-\eta_\Ham r_{\ell k}},
        \quad  
        \left| 
            \frac{\partial \Ham_{\ell k, ab}}{\partial [\bm r_m]_l} 
        \right| 
        + \left| 
            \frac{\partial \Ham_{\ell k, ab}}{\partial \widehat{v}_m} 
        \right| 
        &\lesssim \big(1 + \|\widehat{v}\|_{\ell^\infty}\big)
        e^{-\eta_\Ham [r_{\ell m} + r_{m k}]}
        \label{eq:der-Ham}
    \end{align}
    where $\eta_\Ham$ depends on $\eta_\phi, \eta_\vb$ from \cref{eq:phi-assumption,eq:v-basis-assumption}.
\end{lemma}

\begin{proof}
    Recall, 
    $\Ham_{\ell k, ab} 
    = \int \nabla \phi_{\ell a} \cdot \nabla \phi_{kb} 
    + \sum_{mc} \widehat{v}_{mc} \int \phi_{\ell a} \phi_{kb} \vb_{mc}$. 
    In particular, applying \cref{eq:phi-assumption,eq:v-basis-assumption} together with \cref{lem:elementary}, we obtain the first bound in \cref{eq:der-Ham} with the exponent $\frac{1}{2}\min\{ \eta_\phi, \eta_\vb \}$. 

    On defining  
    $V(x) \coloneqq \sum_{mc} \widehat{v}_{mc} \vb_{mc}$, 
    we have
    \begin{align}
        -\frac{\partial \Ham_{\ell k, ab}}{\partial [\bm r_m]_l} 
        = 
        \widehat{v}_{mc}\int \phi_{\ell a}\phi_{k b} \partial_l \vb_{mc} 
        +
        \begin{cases} 
            \sum_i \int \partial_{il} \phi_{ma} \partial_i\phi_{kb} 
            + \int \partial_l \phi_{ma} \phi_{kb} \, V
            &\text{if } \ell = m \not= k\\
            \sum_i \int \partial_{i} \phi_{\ell a} \partial_{il}\phi_{mb} 
            + \int \phi_{\ell a} \partial_i \phi_{mb} \, V
            &\text{if } \ell \not= m = k\\
            \int \Big[ 
                \sum_i\big( 
                    \partial_{il} \phi_{m a} \partial_i\phi_{mb} + \partial_{i} \phi_{m a} \partial_{il}\phi_{mb} 
                \big)\\
                \qquad\qquad + \big( 
                    \partial_l \phi_{ma} \phi_{mb} + \phi_{ma} \partial_l \phi_{mb}
                \big)  V 
            \Big] 
            &\text{if } \ell = m = k\\
            0 &\textrm{otherwise.}
        \end{cases}
    \end{align}
    Moreover, we have 
    $\frac{\partial \Ham_{\ell k, ab}}{\partial \widehat{v}_{mc}} = \int \phi_{\ell a} \phi_{kb} \vb_{mc}$. 
    Therefore, using the locality of $\phi$ and $\vb$ and their derivatives (\cref{eq:phi-assumption,eq:v-basis-assumption}) together with \cref{lem:elementary}, we have 
    \begin{align}
        \left| 
            \frac{\partial \Ham_{\ell k, ab}}{\partial [\bm r_m]_l} 
        \right| 
        + \left| 
            \frac{\partial \Ham_{\ell k, ab}}{\partial \widehat{v}_m} 
        \right| 
        &\lesssim \big(1 + \|\widehat{v}\|_{\ell^\infty}\big)
        \max\{1,c_{\vb}\}
        c_\phi^2\eta_{\phi}^{-3} 
        e^{-\frac{1}{4}\min\{\eta_{\phi},\eta_{\vb}\} [r_{\ell k} + r_{km} + r_{m\ell}]},
        \label{eq:der-Ham-proof}
    \end{align}
    which concludes the proof.
\end{proof}

Next, we state properties of the Jacobian, $DF_{\beta,\mu}$, of $\Ham \mapsto F_{\beta,\mu}(\Ham)$:
\begin{lemma}
    \label{lem:DF-negative}
    Fix $\beta\in (0,\infty]$, $\mu \in \mathbb R$, and a symmetric matrix $\Ham\in(\mathbb R^{N_\mathrm{b} \times N_\mathrm{b}})^{M\times M}$. In the case $\beta = \infty$, we suppose $\mu\in \conv\big( \sigma(\Ham)\big)  \setminus \sigma(\Ham)$. 
    Then, $DF_{\beta,\mu}(\Ham)$ is negative definite with respect to the Frobenius inner product.
\end{lemma}
\begin{proof}
    This elementary proof uses the same idea as in \cite{Levitt:screening}, and is included here for completeness.

    If $\beta = \infty$, we suppose $\mathscr C$ is a simple closed contour encircling $\sigma(\Ham)\cap (-\infty, \mu)$ and avoiding the rest of the spectrum. On the other hand, if $\beta < \infty$, we suppose $\mathscr C$ is a simple closed contour encircling $\sigma(\Ham)$ and avoiding the poles of $F_{\beta,\mu}$ at $\mu + i \pi\beta^{-1}(2\mathbb Z + 1)$. Then, we have 
    \begin{align}
        DF_{\beta,\mu}( \Ham )_{\ell k ab, \ell^\prime k^\prime a^\prime b^\prime}
        &= \frac{\partial}{\partial \Ham_{\ell^\prime k^\prime, a^\prime b^\prime}}
        \oint_{\mathscr C} 
            F_{\beta,\mu}(z)
            \big( z - \Ham \big)^{-1}_{\ell k, ab} 
        \frac{\mathrm{d}z}{2\pi i} \nonumber\\
        &= \oint_{\mathscr C}  
            F_{\beta,\mu}(z)
            \big( z - \Ham \big)^{-1}_{\ell\ell^\prime,aa^\prime}
            \big( z - \Ham \big)^{-1}_{k^\prime k, b^\prime b}
        \frac{\mathrm{d}z}{2\pi i} \nonumber \\
        &= \sum_{st}
        \oint_{\mathscr C}  
            \frac{F_{\beta,\mu}(z)}{(z - \lambda_s)(z - \lambda_t)}
        \frac
            {\mathrm{d}z}
            {2\pi i} 
        [C^s \otimes C^s]_{\ell\ell^\prime, aa^\prime} 
        [C^t \otimes C^t]_{k^\prime k, b^\prime b}
        \nonumber \\
        &=\sum_{st}
        \frac{F_{\beta,\mu}(\lambda_s) - F_{\beta,\mu}(\lambda_t)}{\lambda_s - \lambda_t}
        [C^s \otimes C^s]_{\ell\ell^\prime, aa^\prime} 
        [C^t \otimes C^t]_{k^\prime k, b^\prime b},
    \end{align}
    with the convention that $\frac{F_{\beta,\mu}(x) - F_{\beta,\mu}(x)}{x-x}\coloneqq F^\prime_{\beta,\ep_{\mathrm{F}}}(x)$ 
    and $\sigma(\Ham) = \{\lambda_1 \leq \dots \leq \lambda_{N_\mathrm{b} \cdot M}\}$. 
    In particular, we have
    \begin{align}
        \Braket{DF_{\beta,\mu}( \Ham ) \Phi, \Phi}
        &= 
        \sum_{\ell\ell^\prime kk^\prime, aa^\prime bb^\prime}
        \Phi_{\ell k,ab}
        DF_{\beta,\mu}( \Ham )_{\ell k ab, \ell^\prime k^\prime a^\prime b^\prime} 
        \Phi_{\ell^\prime k^\prime, a^\prime b^\prime} \nonumber\\
        &= \sum_{st}
        \frac{F_{\beta,\mu}(\lambda_s) - F_{\beta,\mu}(\lambda_t)}{\lambda_s - \lambda_t}
        \big( C^{s\mathrm{T}} \Phi C^t \big)^2
        < 0,
    \end{align}
    which concludes the proof. 
    
    If $\beta = \infty$, then one obtains bounds on the spectrum of $-DF_{\infty, \mu}(\Ham)$:
    \[
        \sigma\big( -DF_{\infty,\mu}(\Ham) \big)
        \subset \Big[\frac{1}{\lambda_{N_\mathrm{b}\cdot M} - \lambda_1}, \frac{1}{\lambda_{N_{\mathrm{el}} + 1} - \lambda_{N_{\mathrm{el}}}}\Big]
    \]
    where $\lambda_{N_{\mathrm{el}}} < \mu < \lambda_{N_{\mathrm{el}} + 1}$.
\end{proof}

We recall the Combes--Thomas \cite{CombesThomas1973} resolvent estimate (the following particular form is taken from \cite{Thomas2020:scTB}, while variants may be found in \cite{ChenOrtner16,ChenLuOrtner18,ChenOrtnerThomas2019:locality}):
\begin{lemma}[Combes--Thomas Resolvent Estimate]
    \label{lem:CT}
    Suppose that $\Ham \in (\mathbb R^{N_{\mathrm{b}} \times N_{\mathrm{b}}})^{M\times M}$ is symmetric with 
    $\big|\Ham_{\ell k, ab}\big| \leq c_{\Ham} e^{-\eta_{\Ham} \, r_{\ell k}}$ 
    (for some $c_\Ham, \eta_{\Ham} > 0$) for all $\ell, k \in \{1,\dots,M\}$ and $a,b \in \{1,\dots,N_{\mathrm{b}}\}$. Then, if $z \in \mathbb C$ with $\mathfrak{d} \coloneqq \mathrm{dist}\big( z, \sigma(\Ham) \big) > 0$, we have 
    \begin{align}
        \label{eq:CT}
        \Big| \big[ (z - \Ham)^{-1} \big]_{\ell k, ab} \Big|
        \leq 2 \mathfrak{d}^{-1} e^{-\eta_{\mathrm{ct}} \, r_{\ell k}} \\
        \text{where} \qquad \eta_{\mathrm{ct}} \coloneqq c \eta_\Ham \min\left\{1, \tfrac{\eta_\Ham^d}{c_\Ham} \mathfrak{d} \right\} \label{eq:CT-exponent}
    \end{align}
    and $c > 0$ depends on $\min\limits_{\ell\not=k} r_{\ell k}$ and $d$.
\end{lemma}

Applying the polynomial approximation result (\cref{poly-approx}), together with the Combes--Thomas estimate, we obtain the following result:
\begin{lemma}
    \label{lem:locality+bodyorder}
    Fix $\bm r$ and $\widehat{v}\in\ell^\infty(I)$, and write $\Ham = H_0 + v$. Suppose that, if $\beta = \infty$, then $\mathsf{g}_{\pm} > 0$. 
    Then, for both $O = F$ and $G$, the polynomial approximations $O_N = F_N$ and $G_N$ as in \cref{poly-approx} satisfy the following: there exists $\theta, \eta_{\mathrm{ct}} > 0$ such that
    \begin{align}
        &\big| D^jO(\Ham)_{\ell kab; \ell_1 k_1 a_1 b_1, \dots, \ell_j k_j a_j b_j} \big| 
        + e^{\theta N}  \big| \big[ D^jO(\Ham) - D^jO_N(\Ham) \big]_{\ell kab; \ell_1 k_1 a_1 b_1, \dots, \ell_j k_j a_j b_j} \big| \nonumber\\
        &\qquad\lesssim e^{-\eta_{\mathrm{ct}} \min\limits_{\sigma \in S_j} [r_{\ell \ell_{\sigma(1)}} + r_{k_{\sigma(1)} \ell_{\sigma(2)}} + \dots + r_{k_{\sigma(j-1)} \ell_{\sigma(j)}} + r_{k_{\sigma(j)} k}]}, 
        \label{eq:lem12}
    \end{align}
    for $0\leq j\leq \nu$. Moreover,  
    $\theta \sim \beta^{-1} + \sqrt{\mathsf{g}_-}\sqrt{\mathsf{g}_+}$ 
    and 
    $\eta_{\mathrm{ct}} \sim \beta^{-1} + \min\{\mathsf{g}_-,\mathsf{g}_+\}$ 
    as $\beta^{-1} + \mathsf{g}\to 0$. If $\beta = \infty$, we have $\theta, \eta \sim \mathsf{g}$ 
    as $\mathsf{g}\to 0$.
\end{lemma}
\begin{proof}
    The Hamiltonian $\Ham \coloneqq H_0 + v$ satisfies the assumption of \cref{lem:CT} with $c_{\Ham} \lesssim 1 + \|\widehat{v}\|_{\ell^\infty}$ and $\eta_\Ham = \frac{1}{2}\eta_{\phi}$. 
    
    Moreover, a simple calculation reveals that, on defining $\mathscr R_z \coloneqq ( z - \Ham )^{-1}$ and thus
    \begin{align}
        [\mathscr R_z^{(j)}]_{\ell kab; \ell_1 k_1 a_1 b_1, \dots, \ell_j k_j a_j b_j}
        \coloneqq \Big[ 
            \sum_{\sigma \in S_j} 
            \mathscr R_z \Delta_{\ell_{\sigma_1} k_{\sigma_1} a_{\sigma_1} b_{\sigma_1}} \mathscr R_z
            \cdots 
             \mathscr R_z \Delta_{\ell_{\sigma_j} k_{\sigma_j} a_{\sigma_j} b_{\sigma_j} }
             \mathscr R_z
        \Big]_{\ell k, ab}, \nonumber
    \end{align}
    where $[\Delta_{\ell k ab}]_{\ell^\prime k^\prime,a^\prime b^\prime} \coloneqq \delta_{\ell\ell^\prime} \delta_{kk^\prime} \delta_{aa^\prime} \delta_{bb^\prime}$, 
    we have 
    \begin{align}
        D^jO(\Ham)
        &= \oint_{\mathscr C} O(z) \mathscr R_z^{(j)} \frac{\mathrm{d}z}{2\pi i},
        \qquad
        D^jO(\Ham) - D^jO_N(\Ham)
        = \oint_{\mathscr C} \big( O - O_N\big)(z) \mathscr R_z^{(j)} \frac{\mathrm{d}z}{2\pi i} \nonumber
    \end{align}
    where $\mathscr C$ is a simple closed positively oriented contour encircling $\sigma(\Ham)$ and avoiding the singularities of $O$. That is, for both $O = F$ and $G$, the contour must be contained in $\mathbb C \setminus \{ \mu + i r \colon r \in \mathbb R, |r| \geq \pi \beta^{-1} \}$.
    
    Therefore, applying \cref{eq:CT}, we obtain
    \begin{align}
        &\Big| D^jO(\Ham)_{\ell kab; \ell_1 k_1 a_1 b_1, \dots, \ell_j k_j a_j b_j} \Big| \nonumber\\
        &\qquad\leq 
        \Big(\frac{2}{\mathfrak{d}}\Big)^{j+1}
        \|O\|_{\mathscr C} 
        \sum_{\sigma\in S_j} e^{-\eta_{\mathrm{ct}} [r_{\ell \ell_{\sigma(1)}} + r_{k_{\sigma(1)} \ell_{\sigma(2)}} + \dots + r_{k_{\sigma(j-1)} \ell_{\sigma(j)}} + r_{k_{\sigma(j)} k}]} \nonumber\\
        &\qquad\leq 
        \Big(\frac{2}{\mathfrak{d}}\Big)^{j+1} 
        j! \,
        \|O\|_{\mathscr C} 
        e^{-\eta_{\mathrm{ct}} \min\limits_{\sigma \in S_j}[r_{\ell \ell_{\sigma_1}} + r_{k_{\sigma_1} \ell_{\sigma_2}} + \dots + r_{k_{\sigma_{j-1}} \ell_{\sigma_j}} + r_{k_{\sigma_j} k}]}
    \end{align}
    where $\mathfrak{d}, \eta_{\mathrm{ct}}$ are the constants from \cref{lem:CT} applied to $\Ham = H_0 + v$, and $\|O\|_{\mathscr C} \coloneqq \len\mathscr C \max\limits_{z\in \mathscr C} |O(z)|$.

    Moreover, we similarly obtain the second inequality in \cref{eq:lem12} by noting that the polynomial approximation result of Lemma~\ref{poly-approx} may be extended to obtain \cite{ThomasChenOrtner2022:body-order}:
    \[
        \| O - O_N \|_{\mathscr C} \lesssim e^{-\theta N}.
    \]
    In fact, one uses the Hermite integral formula \cite{bk:Trefethen2019} and the proof follows in the exact same way as Lemma~\ref{poly-approx}.
\end{proof}

As a direct corollary, we have the following result: 
\begin{corollary}
\label{cor:DjO}
Under the assumptions and notation from \cref{lem:locality+bodyorder}, and for $w^{(1)},\dots,w^{(j)}$ with $|w^{(l)}_{\ell k, ab}| \leq c_l e^{-\eta_l \mathsf d_l(\bm r_\ell,\bm r_k)}$ for some $c_l, \eta_l > 0$, we have 
\begin{align*}
    \big|\big[ 
        D^jO(\Ham) [w^{(1)}, \dots, w^{(j)}]
    \big]_{\ell k, ab}\big|
    &\lesssim_j c_1 \cdots c_j \,
    e^{-\eta \max_l\mathsf{d}_l(\bm r_\ell, \bm r_k)} \\
    \big|\big[ 
        D^j(O-O_N)(\Ham) [w^{(1)}, \dots, w^{(j)}]
    \big]_{\ell k, ab}\big|
    &\lesssim_j c_1 \cdots c_j \, 
    e^{-\theta N} e^{-\eta \max_l \mathsf d_l(\bm r_\ell,\bm r_k)} 
\end{align*}
where 
$\eta = \frac{1}{2}\min\{\eta_{\mathrm{ct}},\eta_1, \dots, \eta_{j}\}$, $\eta_{\mathrm{ct}}$ is the constant from \cref{lem:locality+bodyorder}, and 
$\mathsf d_l$ is symmetric with $r_{\ell k} \leq \mathsf{d}_l(\bm r_\ell, \bm r_{k}) \leq r_{\ell m} + \mathsf{d}_l(\bm r_m, \bm r_{k})$ for all $\ell,k,m$ and $l=1,\dots,j$.

Moreover, in the case 
$w^{(l)} \coloneqq \sum_{mc} \widehat{w}^{(l)}_{mc} \int \phi \otimes \phi  \, \vb_{mc}$ 
for $l=1,\dots,j$ and $0\leq j \leq \nu$, we have 
$c_l = \|\widehat{w}^{(l)}\|_{\ell^\infty}$ and $\eta_l = \frac{1}{2}\eta_\phi$.
\end{corollary}
\begin{remark}
    In the main proofs, we will either apply \cref{cor:DjO} with $\mathsf{d}(\bm r_\ell, \bm r_k) \coloneqq r_{\ell k}$, $\mathsf{d}(\bm r_\ell, \bm r_k) \coloneqq r_{\ell m} + r_{mk}$ for fixed $m$, or $\mathsf{d}(\bm r_\ell, \bm r_k) \coloneqq |\bm r_\ell - x| + |x - \bm r_k|$ for fixed $x$.
\end{remark}    
    
\begin{proof}
    We directly apply \cref{lem:locality+bodyorder} to conclude
    \begin{align}
        &\big|\big[ 
            D^jO(\Ham) [w^{(1)}, \dots, w^{(j)}]
        \big]_{\ell k, ab}\big| \nonumber\\
        &\quad\lesssim c_1\cdots c_j \sum_{\ell_1k_1,\dots,\ell_jk_j}
        e^{-2\eta \min\limits_{\sigma \in S_j}[
            r_{\ell \ell_{\sigma_1}} + \mathsf{d}_{\ell_{\sigma_1}k_{\sigma_1}} + r_{k_{\sigma_1} \ell_{\sigma_2}} + \dots + r_{k_{\sigma_{j-1}} \ell_{\sigma_j}} + \mathsf{d}_{\ell_{\sigma_j}k_{\sigma_j}} + r_{k_{\sigma_j} k}
        ]} \nonumber \\
        &\quad\leq c_1\cdots c_j \Big[ \sum_{\ell_1k_1,\dots,\ell_jk_j}
        e^{-\eta \min\limits_{\sigma \in S_j}[
            r_{\ell \ell_{\sigma_1}} + r_{\ell_{\sigma_1}k_{\sigma_1}} + r_{k_{\sigma_1} \ell_{\sigma_2}} + \dots + r_{k_{\sigma_{j-1}} \ell_{\sigma_j}} + r_{\ell_{\sigma_j}k_{\sigma_j}} + r_{k_{\sigma_j} k}
        ]}
        \Big] 
        e^{-\eta \, \mathsf{d}_{\ell k}}
        \label{eq:cor-proof}
    \end{align}
    where $\mathsf{d}_{\ell k} \coloneqq \mathsf{d}(\bm r_\ell, \bm r_k)$. The summation in the square brackets in \cref{eq:cor-proof} is finite (independent of system size) by a repeated application of \cref{lem:elementary}.

    Moreover, after noting that 
    $\sum_{mc} |\vb_{mc}(x)| 
    \leq c_\vb \sum_{mc} e^{-\eta_{\vb} |x - \bm r_m|}$ 
    is uniformly bounded, we have
    \[
        \left| 
            \sum_{mc} \widehat{w}_{mc} \int \phi_{\ell a}(x) \phi_{kb}(x) \vb_{mc}(x) 
        \right|
        \lesssim \|\widehat{w}\|_{\ell^\infty(I)} 
        \int e^{-\eta_\phi [|x - \bm r_\ell| + |x - \bm r_k|]}
        \lesssim \|\widehat{w}\|_{\ell^\infty(I)} 
        e^{-\frac{1}{2} \eta_{\phi} r_{\ell k}}.
    \]
    Here, we have applied \cref{lem:elementary}. 
\end{proof}    

\subsection{Proofs of the Main Results}

\begin{proof}[Proof of Lemma~\ref{ML-minimiser}]
    We shall consider the finite and zero Fermi-temperature cases together by considering $\mathcal E_\beta$ as defined in \cref{eq:ab-initio-beta} (allowing for the $\beta = \infty$ case). 
    
    Firstly, we note that $\mathcal E_\beta$ is a continuous function of $\widehat{v}$: the mapping taking $\widehat{v}$ to $H_0 + v$ is affine, while the mapping from $H_0 + v$ onto the ordered eigenpairs is continuous. Therefore 
    $\widehat{v} \mapsto \gamma_{\mathrm{TB}}, \rho_{\gamma}$ 
    are both continuous, and thus the composition $\widehat{v} \mapsto \gamma \mapsto \mathcal E_\beta(\gamma)$ is continuous. Moreover, the set of admissible density operators is compact and thus there exists a minimiser $\gamma^0$. 
    
    If $\beta = \infty$ then we assume there is a gap in the spectrum. Therefore, we have 
    $\gamma^0 = \phi^T F_{\beta,\ep_{\mathrm{F}}}(H_0 + v)\phi$ 
    for some $\widehat{v} \in \ell^\infty(I)$. In the zero Fermi-temperature case, $\ep_{\mathrm{F}}$ is any value in the spectral gap, whereas for $\beta < \infty$, the Fermi level $\ep_{\mathrm{F}}$ is the unique solution to 
    $\mathrm{Tr}\,F_{\beta,\ep_{\mathrm{F}}}(H_0 + v) = N_{\mathrm{el}}$. 
    This equation has a unique solution since the mapping 
    $\tau \mapsto \mathrm{Tr}\,F_{\beta,\tau}(H_0 + v)$ 
    is strictly increasing with $\mathrm{Tr}\,F_{\beta,\tau}(H_0 + v) \to 0$ (resp. $M\cdot N_\mathrm{b}$) as $\tau \to -\infty$ (resp. $\tau \to +\infty$).
    
    Rewriting the total energy, we have
    \begin{align} 
        \label{eq:E-rewritten}
        \mathcal E_{\beta}(\gamma) 
        = \sum_{i} \Big[ \lambda_i F_{\beta,\ep_{\mathrm{F}}}(\lambda_i) 
        +\beta^{-1} S\big( F_{\beta,\ep_{\mathrm{F}}}(\lambda_i) \big) \Big]
        &- \sum_{\ell k ab} v_{\ell k, ab} 
        F_{\beta,\ep_{\mathrm{F}}}(H_0 + v)_{\ell k, ab}
        \nonumber\\
        &+ \int \Big[
            \rho_{\gamma} V^{\mathrm{nuc}} + \tfrac{1}{2}\rho_{\gamma} v_{\rm C}\rho_{\gamma} + \ep_{\mathrm{xc}}(\rho_{\gamma}) 
        \Big]
    \end{align}
    where $\sigma(H_0 + v) = \{\lambda_i\}$. Moreover, we fix the numbering of the eigenvalues so that the functions $\widehat{v} \mapsto \lambda_i$ are differentiable for each $i$. 
    
    In order to compute the derivatives of $\mathcal E_\beta$, we require the derivatives of the eigenvalues:
    \begin{align}
        \frac{\partial \lambda_i}{\partial \widehat{v}_{mc}} 
        &= \frac{\partial C^{iT}}{\partial \widehat{v}_{mc}} (H_0 + v) C^i 
        + C^{iT} \frac{\partial v}{\partial \widehat{v}_{mc}} C^i
        + C^{iT} (H_0 + v) \frac{\partial C^{i}}{\partial \widehat{v}_{mc}} \nonumber\\
        &= 2\lambda_i \frac{\partial C^i \cdot C^i}{\partial \widehat{v}_{mc}} 
        + \frac{\partial v}{\partial \widehat{v}_{mc}} \colon  ( C^{i} \otimes C^i ). \label{eq:der-eigen}
    \end{align}
    The first term of \cref{eq:der-eigen} is zero due to the normalisation of the eigenvectors $\|C^i\|_{\ell^2} = 1$. 
    
    Now, we note that at zero Fermi-temperature $\ep_{\mathrm{F}}$ can be fixed constant in a neighbourhood of $\widehat{v}$. At finite Fermi-temperature, we differentiate the constraint $\Tr F_{\beta,\ep_\mathrm{F}}(H_0 + v) = N_{\mathrm{el}}$ and apply \cref{eq:der-eigen}, to obtain
    \begin{align}
        \Tr F_{\beta,\ep_{\mathrm{F}}}^\prime(H_0 + v) \frac{\partial \ep_{\mathrm{F}}}{\partial \widehat{v}_{mc}} 
        &= \sum_i F_{\beta,\ep_{\mathrm{F}}}^\prime(\lambda_i) \frac{\partial \lambda_i}{\partial \widehat{v}_{mc}} 
        = F_{\beta,\ep_{\mathrm{F}}}^\prime(H_0 + v) : \frac{\partial v}{\partial \widehat{v}_{mc}}.
        \label{eq:der-FL-constraint}
    \end{align}
    
    We now may differentiate the first term of \cref{eq:E-rewritten},
    \begin{align}
        &\frac
            {\partial}
            {\partial \widehat{v}_{mc}}
        \sum_i \Big[ \lambda_i F_{\beta,\ep_{\mathrm{F}}}(\lambda_i) 
        + \beta^{-1} S\big( F_{\beta,\ep_{\mathrm{F}}}(\lambda_i) \big) \Big] \nonumber\\
        &=  \Big[ 
            F_{\beta,\ep_{\mathrm{F}}}(H_0 + v) 
            + \Big( (H_0 + v) 
            + \beta^{-1} S^\prime\big( F_{\beta,\ep_{\mathrm{F}}}(H_0 + v) \big) \Big)
            F^\prime_{\beta,\ep_{\mathrm{F}}}(H_0 + v)
        \Big] \colon \frac{\partial v}{\partial \widehat{v}_{mc}} \nonumber\\ 
        &\qquad -
        \sum_i \Big[ 
            \lambda_i + \beta^{-1} S^\prime\big(
                F_{\beta,\ep_{\mathrm{F}}}(\lambda_i)
            \big) 
        \Big]
        F_{\beta,\ep_{\mathrm{F}}}^\prime(\lambda_i) 
        \frac{\partial \ep_{\mathrm{F}}}{\partial \widehat{v}_{mc}} \nonumber\\
        &= \Big[ 
            F_{\beta,\ep_{\mathrm{F}}}(H_0 + v) 
            + \ep_{\mathrm{F}} F^\prime_{\beta,\ep_{\mathrm{F}}}(H_0 + v)
        \Big] \colon \frac{\partial v}{\partial \widehat{v}_{mc}} %
        -
        \ep_{\mathrm{F}}\Tr F^\prime_{\beta,\ep_{\mathrm{F}}}(H_0 + v)
        \frac
            {\partial \ep_{\mathrm{F}}}
            {\partial \widehat{v}_{mc}}  \nonumber\\ 
        &= F_{\beta,\ep_{\mathrm{F}}}(H_0 + v) \colon \frac{\partial v}{\partial \widehat{v}_{mc}}.
        \label{eq:der-E-T1}
    \end{align}
    Here, we have used the fact that
    $S^\prime(x) = \log \frac{x}{1-x}$, 
    and thus 
    $S^\prime\big(F_{\beta,\ep_{\mathrm{F}}}(x)\big) = -\beta (x - \ep_{\mathrm{F}})$. 
    In the final line, we have applied \cref{eq:der-FL-constraint}. 
    
    On the other hand, derivatives of the final term in \cref{eq:E-rewritten} involve derivatives of the electron density: 
    \begin{align}
        \frac
            {\partial [\gamma_{\mathrm{TB}}]_{\ell k, ab}}
            {\partial \widehat{v}_{mc}}
        &= \Big[ 
            DF_{\beta,\ep_{\mathrm{F}}}( H_0 + v )
            \frac
                {\partial v}
                {\partial \widehat{v}_{mc}}
        \Big]_{\ell k, ab} 
        \qquad \text{and thus} \nonumber\\
        \frac
            {\partial \rho_{\gamma}(x)}
            {\partial \widehat{v}_{mc}}
        &= \Braket{
            DF_{\beta,\ep_{\mathrm{F}}}( H_0 + v ) \phi(x) \otimes \phi(x),
            \frac
                {\partial v}
                {\partial \widehat{v}_{mc}}
        }.
        \label{eq:der-rho}
    \end{align}
    In particular, combining \cref{eq:der-E-T1,eq:der-rho}, we obtain
    \begin{align}
        0 = \frac{\partial \mathcal E_{\beta}(\gamma)}{\partial \widehat{v}_{mc}} 
        &= 
        - \sum_{\ell k, ab} v_{\ell k,ab} 
        \Big[
            DF_{\beta,\ep_{\mathrm{F}}}(H_0 + v) \frac{\partial v}{\partial \widehat{v}_{mc}}
        \Big]_{\ell k,ab}
        +
        \int 
        \frac
            {\partial \rho_{\gamma}(x)}
            {\partial \widehat{v}_{mc}}
        V_{\mathrm{eff}}[\rho_\gamma](x) \mathrm{d}x \nonumber\\
        &= \Braket{
        -DF_{\beta,\ep_{\mathrm{F}}}( H_0 + v ) \Big[
            v - \int \phi \otimes \phi \, V_{\mathrm{eff}}[\rho_\gamma]
        \Big], \frac
            {\partial v}
            {\partial \widehat{v}_{mc}}
        },
        \label{eq:der-E-der-v}
    \end{align}
    concluding the proof of \cref{eq:ML-EL}.

    We move on to consider \cref{eq:ML-approx}. To simplify notation, we write 
    $g \coloneqq \sum_{mc} \widehat{v}_{mc} \vb_{mc} - V_{\mathrm{eff}}[\rho_\gamma]$. First, we note by \cref{lem:DF-negative}, that
    $-DF_{\beta,\ep_{\mathrm{F}}} 
    \coloneqq -DF_{\beta,\ep_{\mathrm{F}}} (H_0 + v)$
    is positive definite and thus has a positive square root 
    $(-DF_{\beta,\ep_{\mathrm{F}}})^{1/2}$ 
    with inverse 
    $(-DF_{\beta,\ep_{\mathrm{F}}})^{- 1/2}$. 
    Then, by \cref{eq:der-E-der-v}, we have 
    \begin{align}
        &\left\|
            (-DF_{\beta,\ep_{\mathrm{F}}})^{-1/2} 
            (-DF_{\beta,\ep_{\mathrm{F}}})^{1/2}
            \int \phi \otimes \phi \, g
        \right\|^2 \nonumber \\
        &\qquad\leq \left\| (-DF_{\beta,\ep_{\mathrm{F}}})^{-1} \right\| 
        \Braket{
            -DF_{\beta,\ep_{\mathrm{F}}}
            \int \phi \otimes \phi \, g,
            \int \phi \otimes \phi \, g
        }  \nonumber \\
        &\qquad \leq 
        \left\| 
            (-DF_{\beta,\ep_{\mathrm{F}}})^{-1} 
        \right\|   
        \min_{g_{\vb} \in \mathrm{span} \, \vb} 
        \Braket{
            -DF_{\beta,\ep_{\mathrm{F}}}
            \int \phi \otimes \phi \, g,
            \int \phi \otimes \phi \, (g-g_{\vb})
        }  \nonumber \\
        &\qquad\leq 
        \left\| 
            (-DF_{\beta,\ep_{\mathrm{F}}})^{-1} 
        \right\|  
        \big\| 
            -DF_{\beta,\ep_{\mathrm{F}}} 
        \big\|  
        \left\|
            \int \phi \otimes \phi \, g
        \right\|
        \min_{g_{\vb} \in \mathrm{span}\, \vb}\left\|
            \int \phi \otimes \phi \, (g - g_{\vb})
        \right\|
    \end{align}
    which concludes the proof of \cref{eq:ML-approx}.
    
    The analogous results in the grand-canonical ensemble can be shown in the exact same way (but are simpler).  
\end{proof}

\begin{proof}[Proof of Theorem~\ref{locality}]
    We consider $\mathcal N_\ell \subset \mathbb R^3$ such that 
    $\mathbb R^3 = \bigcup_{\ell = 1}^M \mathcal N_\ell$ 
    with $(\mathcal N_\ell)_\ell$ pairwise disjoint and  
    $\mathrm{dist}(\bm r_m, \mathcal N_\ell) \geq \frac{1}{2}r_{\ell m}$ for all $\ell, m$. 
    Then, we may define
    \begin{align}
        \rho_\ell(\bm r, \widehat{v};x) 
        &= 2\sum_{k} \phi_{\ell}(x)^T F(H_0 + v)_{\ell k} \phi_{k}(x), 
        \qquad \text{and}
        \label{eq:locality-rho} \\
        E_\ell(\bm r, \widehat{v}) 
        &= \mathrm{tr} \Big[ 
            (H_0-\mu) F(H_0 + v)
            + \beta^{-1} S\big( F(H_0 + v) \big)
        \Big]_{\ell\ell} + \int_{\mathcal N_\ell} \ep_\mathrm{xc}( \rho(\bm r, \widehat{v}) ). 
        \label{eq:locality-E} 
    \end{align}
    Therefore, we have \cref{eq:locality-decomp} with 
    $E_{\mathrm{el}}[\rho] = \int \big[
        \rho V^\mathrm{nuc} + \tfrac{1}{2} \rho v_{\rm C} \rho
    \big]$. In the zero Fermi-temperature case, we may neglect the entropy contribution altogether, whereas, at finite Fermi-temperature, we have
    $xF(x) + \beta^{-1}S(F(x)) = \mu F(x) + \beta^{-1}\log\big( 1 - F(x) \big)$, and thus
    \begin{align}
        E_\ell(\bm r, \widehat{v}) = \mathrm{tr} \Big[ 
            G(H_0 + v)
            - v F(H_0 + v)
        \Big]_{\ell\ell} + \int_{\mathcal N_\ell} \ep_\mathrm{xc}( \rho(\bm r, \widehat{v}) ) 
    \end{align}
    where $G(x) \coloneqq \beta^{-1} \log\big( 1-  F(x) \big)$. In the zero Fermi-temperature limit, one recovers \cref{eq:locality-E}.
    
    Therefore, in order to prove the locality estimates, we need to show that derivatives of the functions $x \mapsto F(x)$ and $x \mapsto G(x)$ evaluated at $H_0 + v$ are exponentially localised. To do this, we simply apply \cref{lem:locality+bodyorder}. 
    
    In previous works \cite{ChenOrtnerThomas2019:locality,Thomas2020:scTB, ChenOrtner16,ChenLuOrtner18}, the Combes--Thomas estimate (\cref{lem:CT}) is applied to obtain the exponential localisation of the local observables $O(\Ham)_{\ell\ell}$ where $x \mapsto O(x)$ is analytic in a neighbourhood of $\sigma(\Ham)$. Here, we extend the analysis to the off-diagonal entries $O(\Ham)_{\ell k}$. 
    
    To simplify notation in the following, we write
    $
        \frac
            {\partial}
            {\partial (\bm r_m, \widehat{v}_m)}
        \coloneqq \big( 
        \frac
            {\partial}
            {\partial \bm r_m}, 
        \frac
            {\partial}
            {\partial \widehat{v}_m}
        \big)^T.
    $
    Since, by \cref{lem:Hamiltonian} we have 
    $\left|
        \frac
            {\partial \Ham_{\ell k,ab}}
            {\partial (\bm r_m, \widehat{v}_m)}
    \right| 
    \lesssim ( 1 + \|\widehat{v}\|_{\ell^\infty} ) e^{-\eta_\Ham \mathsf{d}_{\ell k}}$ 
    where $\mathsf{d}_{\ell k} \coloneqq r_{\ell m} + r_{mk}$, we apply \cref{cor:DjO} to conclude
    \begin{align}
        \left|\frac{\partial O(\Ham)_{\ell k,ab}}{\partial (\bm r_m, \widehat{v}_m)}\right|
        &= 
        \left|\Big[
            DO(\Ham)
            \frac
                {\partial \Ham}
                {\partial (\bm r_m, \widehat{v}_{m})} 
        \Big]_{\ell k, ab}\right|
        \lesssim 
        \big(1 + \|\widehat{v}\|_{\ell^\infty}\big)
        e^{-\frac{1}{2}\min\{\eta_{\mathrm{ct}},\eta_\Ham\} \, \mathsf{d}_{\ell k}},
        \label{eq:local-sketch}
    \end{align}
    where $\eta_{\mathrm{ct}}$ is the constant from \cref{lem:locality+bodyorder}.

    For the remainder, to simplify the notation, we will write 
    $\big| O(\Ham)_{\ell k}\big|\leq c_0 e^{-\eta_0 r_{\ell k}}$ and 
    $\big|
        \frac
            {\partial O(\Ham)_{\ell k}}
            {\partial (\bm r_m, \widehat{v}_m)}
    \big| 
    \leq c_1 e^{-\eta_1[r_{\ell m} + r_{mk}]}$ 
    for 
    $O \in \{ x \mapsto x, F, G\}$ 
    and $c_0,c_1,\eta_0,\eta_1 > 0$. The proofs for higher derivatives (assuming $\nu$ as in \cref{eq:phi-assumption,eq:v-basis-assumption} is sufficiently large) are analogous: we conclude by using the exponential localisation of the Hamiltonian and \cref{cor:DjO}, together with
    \begin{align}
        \frac
            {\partial^2 O(\Ham)}
            {\partial (\bm r_{m_1}, \widehat{v}_{m_1}) \partial (\bm r_{m_2}, \widehat{v}_{m_2}) }
        &= DO(\Ham)
        \frac
            {\partial^2 \Ham}
            {\partial (\bm r_{m_1}, \widehat{v}_{m_1}) \partial (\bm r_{m_2}, \widehat{v}_{m_2}) } \nonumber\\
        &\qquad+ D^2O(\Ham) \Big[ 
            \frac
            {\partial \Ham}
            {\partial (\bm r_{m_1}, \widehat{v}_{m_1})}, 
            \frac
            {\partial \Ham}
            {\partial (\bm r_{m_1}, \widehat{v}_{m_1})}
        \Big].
    \end{align}
    
    Returning to the electron density \cref{eq:locality-rho}, we therefore obtain
    \begin{align}
        \left| 
            \frac
                {\partial \rho_\ell(\bm r, \widehat{v};x)}
                {\partial (\bm r_m, \widehat{v}_m)}
        \right|
        &= \left| 
        2\sum_{k,ab} \phi_{\ell a}(x) \phi_{k b}(x) 
        \frac
            {\partial F(H_0 + v)_{\ell k, ab}}
            {\partial (\bm r_m, \widehat{v}_m)}
        \right| \nonumber \\
        &\leq 2c_\phi^2 c_{1} \sum_{k} e^{-\eta_\phi [|x - \bm r_\ell| + |x - \bm r_k|]} e^{-\eta_1 [r_{\ell m} + r_{km}]} \nonumber \\
        &\lesssim e^{-[ \eta_\phi |x - \bm r_\ell| + \eta_1 r_{\ell m} + \frac{1}{2} \min\{ \eta_\phi, \eta_1 \} |x - \bm r_m|]}.
    \end{align} 
 
    Therefore, derivatives of the site energies have analogous locality properties. In the $\beta = \infty$ case, we have
    \begin{align}
        \frac
            {\partial E_\ell(\bm r, \widehat{v})}
            {\partial (\bm r_m, \widehat{v}_m)}
        &= \sum_{k,ab} \Big[
        \frac
            {\partial [H_0]_{\ell k, ab}}
            {\partial (\bm r_m, \widehat{v}_m)} 
        F(H_0 + v)_{\ell k,ab} 
        + (H_0 - \mu)_{\ell k, ab} 
        \frac
            {\partial F(H_0 + v)_{\ell k,ab}}
            {\partial (\bm r_m, \widehat{v}_m)} \Big]
        \nonumber \\
        &\qquad + 
        \int_{\mathcal N_\ell} 
        \ep_{\mathrm{xc}}^\prime\big( \rho(\bm r, \widehat{v}) \big) 
        \frac
            {\partial \rho(\bm r, \widehat{v})}
            {\partial (\bm r_m, \widehat{v}_m)},
    \end{align}
    whereas, for $\beta < \infty$, we have
    \begin{align}
        \frac
            {\partial E_\ell(\bm r, \widehat{v})}
            {\partial (\bm r_m, \widehat{v}_m)}
        &= \tr 
        \frac
            {\partial G(H_0 + v)_{\ell \ell, aa}}
            {\partial (\bm r_m, \widehat{v}_m)} 
        - \sum_{k,ab} \Big[
        \frac
            {\partial v_{\ell k, ab}}
            {\partial (\bm r_m, \widehat{v}_m)} 
        F(H_0 + v)_{\ell k, ab}
        + v_{\ell k, ab}
        \frac
            {\partial F(H_0 + v)_{\ell k,ab}}
            {\partial (\bm r_m, \widehat{v}_m)} 
        \Big]
        \nonumber \\
        &\qquad + 
        \int_{\mathcal N_\ell} 
        \ep_{\mathrm{xc}}^\prime\big( \rho(\bm r, \widehat{v}) \big) 
        \frac
            {\partial \rho(\bm r, \widehat{v})}
            {\partial (\bm r_m, \widehat{v}_m)}.
    \end{align}
    In both cases, we therefore have 
    \begin{align}
        \bigg|
        \frac
            {\partial E_\ell(\bm r, \widehat{v})}
            {\partial (\bm r_m, \widehat{v}_m)}
        \bigg| 
        &\lesssim e^{-\eta_1 r_{\ell m}} 
        + \sum_k e^{-\eta_1 [r_{\ell m} + r_{mk}]} e^{-\eta_0 r_{\ell k}}
        + \int_{\mathcal N_\ell} 
        e^{-\frac{1}{2} \min\{ \eta_\phi, \eta_1 \} |x - \bm r_m|} \mathrm{d}x\nonumber\\
        &\lesssim e^{-\frac{1}{2} \min\{ \eta_\phi, \eta_1 \} \mathrm{dist}(\bm r_m, \mathcal N_\ell)}
        \leq e^{-\frac{1}{4} \min\{ \eta_\phi, \eta_1 \} r_{\ell m}},
    \end{align}
    which concludes the proof. Here, we have used the assumption  $|\ep_{\mathrm{xc}}^\prime(\rho)| \lesssim |\rho|^{1/3}$ and 
    $|\rho(\bm r, \widehat{v};x)| \leq 
    N_{\mathrm{b}}^2 c_\phi^2 c_1 
    \sum_{\ell k} 
        e^{-\eta_\phi [ |x - \bm r_\ell| + |x - \bm r_k| ] } e^{-\eta_0 r_{\ell k}}
    \lesssim 1$.
\end{proof}

\begin{remark}
    As mentioned in Remark~\ref{rem:locality-N}, the body-ordered approximations as defined in \cref{eq:decomp-N} inherit the same locality properties as in Theorem~\ref{locality}. To show this, one may follow the same proof with $\big| O_N(\Ham)_{\ell k}\big|\leq c_0 e^{-\eta_0 r_{\ell k}}$ and 
    $\big|
        \frac
            {\partial O_N(\Ham)_{\ell k}}
            {\partial (\bm r_m, \widehat{v}_m)}
    \big| 
    \leq c_1 e^{-\eta_1[r_{\ell m} + r_{mk}]}$ 
    for 
    $O_N \in \{F_N, G_N\}$ 
    and some $c_0,c_1,\eta_0,\eta_1 > 0$ (which follows directly from Lemma~\ref{lem:locality+bodyorder}). 
\end{remark}

\begin{proof}[Proof of Theorem~\ref{body-order}]
    Since $v_{\ell k, ab} = \sum_{mc} \widehat{v}_{mc} \int \phi_{\ell a} \, \vb_{mc} \, \phi_{kb}$, on writing $\Ham \coloneqq H_0 + v$, we have
        \begin{align}
            \Ham_{\ell k,ab} &= 
            \int \nabla\phi_{Z_\ell a} \cdot \nabla\phi_{Z_k b}(\,\cdot\, - \bm r_{\ell k}) +
            \int \phi_{Z_\ell a} \phi_{Z_k b}(\,\cdot\, - \bm r_{\ell k}) 
            \Big[ 
                \widehat{v}_{\ell} \cdot  \vb + \widehat{v}_{k} \cdot \vb(\,\cdot\, - \bm r_{\ell k}) 
            \Big] \nonumber \\
            &\qquad + \sum_{m\not\in \{\ell,k\}}  \int \phi_{Z_\ell a}(\,\cdot\,+\bm r_{\ell m}) \phi_{Z_kb}(\,\cdot\, + \bm r_{km}) \, \widehat{v}_{m} \cdot\vb,
            \label{eq:Ham-3-body}
        \end{align}
        a quantity of body-order at most $3$. In particular, we may write $\Ham_{\ell k} = \sum_m \Ham_{\ell k m}$ where $\Ham_{\ell k m}$ has body-order at most $3$ (in the combined variables $\{(\bm r_m, \widehat{v}_m)\}_{m}$). 
        Therefore, polynomials of the Hamiltonian also have finite body-order:
        \begin{align}
            [\Ham^N]_{\ell k} 
            &= \sum_{\ell_{1}, \dots, \ell_{N-1}} \Ham_{\ell \ell_1} \Ham_{\ell_1 \ell_2} \cdots \Ham_{\ell_{N-1} k} 
            = \sum_{\above{\ell_{1}, \dots, \ell_{N-1}}{m_1,\dots,m_{N}}} \Ham_{\ell \ell_1 m_1} \Ham_{\ell_1 \ell_2 m_2} \cdots \Ham_{\ell_{N-1} k m_N},
        \end{align}
        a quantity of body-order at most $2N+1$. That is, there exists $\bm U_N^{(1)}$ as in \cref{eq:U} such that $F_N(H_0 + v) = \bm U_N^{(1)}$. See \cite{ThomasChenOrtner2022:body-order}, for an explicit definition of the $U_{nN}^{(1)}$. 
        In particular, we have \cref{eq:rho-phi-U-phi}, and thus the nonlinearity in \cref{eq:E-tr+ep} is given by 
        \[
            \ep_\ell(\bm U) \coloneqq \int_{\mathcal N_\ell} \ep_{\mathrm{xc}}\big( \bm U \colon [\phi(x) \otimes \phi(x)] \big).
        \]
        Since both $F_N(\Ham)$ and $G_N(\Ham)$ have body-order at most $2N+1$, the quantities 
        \begin{gather}
            \tr\big[ 
                (H_0 - \mu) F_{N}(H_0 + v) 
            \big]_{\ell \ell}, 
            \qquad \text{and} \qquad 
            \tr \big[ G_N(H_0 + v) - vF_N(H_0 + v) \big]_{\ell\ell}
        \end{gather}
        are both of body-order at most $2N+1$.         
    \end{proof}

\begin{proof}[Proof of Lemma~\ref{ML-minimiser-N}]
    The proof follows the exact same argument as in \cref{ML-minimiser}, where, in the finite Fermi-temperature case, we require $G_N^\prime = F_N$.
\end{proof}

\begin{proof}[Proof of Theorem~\ref{minimisers}]
\textit{Preliminaries.} Define $T_N\colon \ell^\infty(I) \to \ell^\infty(I)$ where the index set is defined as $I \coloneqq \{ (m,c) \}$ by 
\begin{align}
    \label{eq:TN}
    T_N(\widehat{v})_{mc} \coloneqq \Braket{ - D F_N\big( H_0 + v \big) 
    \Big[ v - \int \phi \otimes \phi \, V_{\mathrm{eff}}[\rho_N]  \Big], \int \phi \otimes \phi \, \vb_{mc} }
\end{align}   
where $v \coloneqq \sum_{mc} \widehat{v}_{mc} \int \phi \otimes \phi \, \vb_{mc}$ and $\rho_N(x) \coloneqq \phi(x)^T F_N(H_0 + v) \phi(x)$. Moreover, recall $V_{\mathrm{eff}}[\rho] \coloneqq V^{\mathrm{nuc}} + v_{\rm C} \rho + \ep^\prime_{\mathrm{xc}}(\rho)$.
 
We apply the inverse function theorem on $T_N$ about $\widehat{v^\star}$: if there exist $\delta, L, \ep_N, c_{\mathrm{stab},N} > 0$ such that 
\begin{enumerate}[label=(\textit{\roman*})]
    \item $T_N \colon \ell^\infty(I) \to \ell^\infty(I)$ continuous, continuously differentiable on $B_{\delta}(\widehat{v^\star})$,
    
    \item Lipschitz: $\| DT_N(\widehat{v_1}) - DT_N(\widehat{v_2}) \| \leq L\|\widehat{v_1} - \widehat{v_2}\|$ for all $\widehat{v_1}, \widehat{v_2} \in B_\delta(\widehat{v^\star})$,
    
    \item Consistency: $\| T_N(v^\star)\| \leq \ep_N$,   
    
    \item Stability: $DT_N(\widehat{v^\star})$ is non-singular with $\|DT_N(\widehat{v^\star})^{-1}\| \leq c_{\mathrm{stab},N}$,
        
    \item  Newton–Kantorovich condition: $2 c_{\mathrm{stab},N}^2 \ep_N L < 1$,
    \item $2 c_{\mathrm{stab},N} \ep_N < \delta$,
\end{enumerate}
then the Newton iteration starting at $\widehat{v^\star}$ is well-defined and converges quadratically to $\widehat{v^\star_N}$, a unique solution to $T_N = 0$ on $\overline{B_{\|\cdot\|_{\ell^\infty} }(\widehat{v^\star};2 c_{\mathrm{stab},N} \ep_N)}$ \cite{Zhengda1993:NewtonIteration,Kantorovich1948}. 

Define $T\colon \ell^\infty(I) \to \ell^\infty(I)$ by
\begin{align}
    \label{eq:T}
    T(\widehat{v})_{mc} \coloneqq \Braket{ - D F\big( H_0 + v \big) 
    \Big[ v - \int \phi \otimes \phi \, V_{\mathrm{eff}}[\rho]  \Big], \int \phi \otimes \phi \, \vb_{mc} }, 
\end{align} 
where $\rho(x) \coloneqq \phi(x)^T F(H_0 + v) \phi(x)$. In order to simplify the presentation, we will first show that \textit{(i)--(vi)} result from the following: there exist $c_0, c_1, c_2, c_3, c_4 > 0$ such that
\begin{align}
    \label{eq:T-smooth}
    T\colon \ell^\infty(I) \to \ell^\infty(I) \text{ is smooth on } B_\delta(\widehat{v^\star}) \text{ with } \delta \coloneqq c_0 M^{-\frac{1}{2}},
\end{align}
and, for all $\widehat{v} \in B_\delta(\widehat{v^\star})$, we have
\begin{align}
    \label{eq:a}
    \|DT_N(\widehat{v})\|_{\ell^\infty \to \ell^\infty} &\leq c_1 M, \\
    \label{eq:b}
    \|D^2T_N(\widehat{v})\|_{\ell^\infty\times\ell^\infty \to \ell^\infty} &\leq c_2 M, \\
    \label{eq:c}
    \|(T-T_N)(\widehat{v})\|_{\ell^\infty} &\leq c_3 M e^{-\theta N}, \\
    \label{eq:d}
    \|D(T-T_N)(\widehat{v})\|_{\ell^\infty \to \ell^\infty} &\leq c_4 M e^{-\theta N}.
\end{align}
The proof of \cref{eq:T-smooth,eq:a,eq:b,eq:c,eq:d} will follow the conclusion of the proof of \cref{minimisers}.

\textit{(i)} Since $F_N$ is a polynomial, and $\widehat{v} \mapsto H_0 + v$ is smooth, $T_N$ is smooth on $\ell^\infty(I)$.

\textit{(ii)} For fixed $\widehat{v_0}, \widehat{v_1} \in B_{\|\cdot\|_{\ell^\infty}}(\widehat{v^\star}; \delta)$, we define $\widehat{v_t} \coloneqq t\widehat{v_1} + (1 - t) \widehat{v_0}$, and use \cref{eq:b} to conclude
\begin{align}
    %
    %
    %
    %
    \|DT_N(\widehat{v_0}) - DT_N( \widehat{v_1}) \|_{\ell^\infty \to \ell^\infty} 
    %
    %
    &\leq \int_0^1 
    \left\| 
        D^2T_N(\widehat{v_t})
    \right\|_{\ell^\infty \times \ell^\infty \to \ell^\infty}
    \mathrm{d}t \, \| \widehat{v_0} - \widehat{v_1} \|_{\ell^\infty} \nonumber\\
    &\leq c_2 M \| \widehat{v_0} - \widehat{v_1} \|_{\ell^\infty}
    \eqqcolon L \| \widehat{v_0} - \widehat{v_1} \|_{\ell^\infty}. 
\end{align}

\textit{(iii)} Consistency: by \cref{eq:c}, we have 
$\|T_N(\widehat{v^\star})\|_{\ell^\infty} 
= \|(T-T_N)(\widehat{v^\star})\|_{\ell^\infty} 
\leq c_3 M e^{-\theta N} \eqqcolon \ep_N$. 

\textit{(iv)} Stability: We have $\| DT(\widehat{v^\star})^{-1}\| \leq c_{\mathrm{stab}}$. Therefore, for $c_4 M e^{-\theta N} \leq \tfrac{1}{2 c_{\mathrm{stab}}}$, we apply \cref{eq:d} to conclude 
\begin{align}
    \| D(T-T_N)(\widehat{v^\star})\|_{\ell^\infty \to \ell^\infty} 
    &\leq c_4 M e^{-\theta N}
    \leq\frac{1}{2}
    \| DT(\widehat{v^\star})^{-1}\|_{\ell^\infty \to \ell^\infty}^{-1}.
    %
\end{align}
Therefore, for such $N$, $DT_N(\widehat{v^\star})$ is invertible with 
\[
    \|DT_N(\widehat{v^\star})^{-1}\|_{\ell^\infty\to\ell^\infty} \leq 2 \| DT(\widehat{v^\star})^{-1}\|_{\ell^\infty \to \ell^\infty} \leq 2 c_{\mathrm{stab}} 
    \eqqcolon c_{\mathrm{stab},N}.
\]

\textit{(v), (vi)} With the choices of $\delta, L, \ep_N, c_{\mathrm{stab},N}>0$ as above, we require
\begin{align}
    2 c_{\mathrm{stab},N}^2 \ep_N L &=
    8 c_2 c_3 c_{\mathrm{stab}}^2 M^2 e^{-\theta N} \leq 1
    \\
    2 c_{\mathrm{stab},N} \ep_N &= 
    4c_3 c_{\mathrm{stab}} Me^{-\theta N} < \delta = c_0 M^{-\frac{1}{2}}.
\end{align}

\textit{Conclusion:} Therefore, there exists $\overline{C} = \overline{C}(c_{\mathrm{stab}}, c_0, c_1, c_2, c_3, c_4) > 0$ such that if
\[
    \theta N > 2 \log M  + \overline{C},
\]
then there exists $\widehat{v^\star_N}$ such that $T_N(\widehat{v^\star_N}) = 0$ and 
\[
    \|\widehat{v^\star_N} - \widehat{v^\star}\|_{\ell^\infty} \leq 2 c_{\mathrm{stab},N} \ep_N 
    \leq 4 c_3 c_{\mathrm{stab}} M e^{-\theta N}. 
\]

\textit{Density:} We have $|\rho_N(\widehat{v}^\star_N;x) - \rho(\widehat{v}^\star;x)| \lesssim e^{-\theta N} + |\rho(\widehat{v}^\star_N;x) - \rho(\widehat{v}^\star;x)|$ and, using the same argument as in \cref{lem:locality+bodyorder}, we conclude
\begin{align}
    |\rho(\widehat{v}^\star;x) - \rho(\widehat{v}^\star_N;x)| 
    &\lesssim \sum_{\ell k} e^{-\eta_\phi [|x - \bm r_\ell| + |x - \bm r_k|]}
    \big|\big[ F(H_0 + v^\star) - F(H_0 + v^\star_N) \big]_{\ell k}\big| \nonumber \\
    &\lesssim \sum_{\ell k \ell^\prime k^\prime} \|F\|_{\mathscr C} e^{-\eta_\phi [|x - \bm r_\ell| + |x - \bm r_k|]} e^{-\eta_{\mathrm{ct},N} r_{\ell\ell^\prime}} |[v^\star_N - v^\star]_{\ell^\prime k^\prime}| e^{-\eta_{\mathrm{ct}} r_{k^\prime k}} \nonumber \\
    &\lesssim e^{-\eta_\phi \min_\ell|x - \bm r_\ell|} \| \widehat{v}^\star - \widehat{v}^\star_N \|_{\ell^\infty} 
    \lesssim M e^{-\theta N} \label{eq:density-pf}
\end{align}
where $\eta_{\mathrm{ct}}, \eta_{\mathrm{ct},N}$ are the constants from \cref{lem:CT} when applied to $H_0 + v^\star$ and $H_0 + v^\star_N$, respectively, and we have applied \cref{eq:V-bound}, below.

\textit{Energy:} By applying \cref{lem:locality+bodyorder} and similar arguments as in \cref{eq:density-pf}, we have
\begin{align}
    \big|\big[ O_N(H_0 + v^\star_N) - O(H_0 + v^\star) \big]_{\ell k}\big| 
    &\lesssim 
    e^{-\theta N} e^{-\eta_{\mathrm{ct}} r_{\ell k}} 
    + \big|\big[ O(H_0 + v^\star_N) - O(H_0 + v^\star) \big]_{\ell k}\big| \nonumber \\
    &\lesssim e^{-\theta N} e^{-\eta_{\mathrm{ct}} r_{\ell k}} 
    + \| \widehat{v}^\star_N - \widehat{v}^\star \|_{\ell^\infty} 
    e^{-\eta r_{\ell k}} \nonumber \\
    &\lesssim M e^{-\theta N} e^{-\eta r_{\ell k}} \label{eq:ONN-O}
\end{align}
where $\eta \coloneqq \frac{1}{4} \min\{\eta_{\mathrm{ct},N}, \eta_{\mathrm{ct}}, \eta_\phi\}$. In particular, one may show that 
\begin{align}
    \big|\mathcal G_N(\widehat{v}^\star_N) - \mathcal G(\widehat{v}^\star)\big|
    &\lesssim M^2 e^{-\theta N} 
    + \int 
    \big| 
        \ep_\mathrm{xc}^\prime\big(\xi(x)\big) 
    \big|  
    \big| 
        \rho_N(\widehat{v}^\star_N;x) - \rho(\widehat{v}^\star;x) 
    \big| \mathrm{d}x \nonumber \\
    &\lesssim M^2 e^{-\theta N}
    + \int e^{-\eta_\phi \min_\ell |x - \bm r_\ell|} \mathrm{d}x \, Me^{-\theta N}
    \lesssim M^2 e^{-\theta N}
\end{align}
where $\xi(x) \in \big[ \rho_N(\widehat{v}^\star_N;x), \rho(\widehat{v}^\star;x) \big]$ and, in the final line, we have applied \cref{eq:density-pf}.
\end{proof}

\subsection{Proof of (\ref{eq:T-smooth})--(\ref{eq:d})}

 For $\widehat{v} \in \ell^\infty(I)$, we define 
    $V(x) \coloneqq \sum_{mc} \widehat{v}_{mc} \vb_{mc}(x)$ and 
    $v = \int \phi \otimes \phi \, V$, 
    $\rho(x) \coloneqq \phi(x)^T F(\Ham) \phi(x)$,  
    and $\rho_N(x) \coloneqq \phi(x)^T F_N(\Ham) \phi(x)$ 
    where $\Ham \coloneqq H_0 + v$.

\begin{proof}[Proof of (\ref{eq:T-smooth})]
We first note that 
\begin{align}
    \mathrm{dist}\big( 
        \sigma(H_0 + v_1),
        \sigma(H_0 + v_2)
    \big)
    &\leq \|v_1 - v_2\|_{\ell^2 \to \ell^2} \nonumber\\ 
    &\leq \sup_{\|w\|_{\ell^2} = 1} \sqrt{\sum_{\ell a} \Big[
        \sum_{kb,mc} \int \phi_{\ell a} \phi_{kb} [\widehat{v_1} - \widehat{v_2}]_{mc}\vb_{mc} w_{kb}
    \Big]^2}\nonumber\\
    &\leq C_p \sqrt{M} \|\widehat{v_1} - \widehat{v_2}\|_{\ell^p} \label{eq:dist-spec}
\end{align}
where
\begin{align*}
    &C_p^2 \coloneqq \sup_{\|w\|_{\ell^2} = 1} \sup_\ell \Big[ \sum_{k,ab} 
        \int \phi_{\ell a}(x) \phi_{kb}(x) w_{kb} \|\vb(x)\|_{\ell^q} \mathrm{d}x 
    \Big]^2
\end{align*}
and $1\leq p,q,\leq \infty$ with $\frac{1}{p} + \frac{1}{q} = 1$.

Recall that $\mathsf g(v) \coloneqq \max \big(\sigma( H_0 + v ) \cap (-\infty, \mu]\big) - \min\big(\sigma( H_0 + v ) \cap [\mu,\infty)\big)$. 
By \cref{eq:dist-spec}, we have $\mathsf{g}(v) \geq \frac{1}{2}\mathsf{g}(v^\star)$ for all 
$\widehat{v} \in B_{\|\cdot\|_{\ell^\infty}}(\widehat{v^\star}; \delta)$ where $\delta \coloneqq c_0 M^{-\frac{1}{2}}$ and $c_0 \coloneqq \frac{\mathsf{g}(v^\star)}{4 C_\infty}$. 
Therefore, on $B_{\|\cdot\|_{\ell^\infty}}(\widehat{v^\star}; \delta)$, the mapping $\widehat{v} \mapsto F(H_0 + v)$ is smooth and thus $\widehat{v} \mapsto T(\widehat{v})$ is smooth.
\end{proof}

In order to prove \cref{eq:a,eq:b,eq:c,eq:d}, we require the following basic facts:
\begin{align}
    | v_{\ell k, ab} | &\lesssim \|\widehat{v}\|_{\ell^\infty(I)} e^{-\frac{1}{2} \eta_{\phi} r_{\ell k}}\label{eq:V-bound} 
\end{align}
\begin{proof} 
$\big| 
    v_{\ell k,ab} 
\big| 
\leq \sum_{mc} c_\phi^2 c_\vb |\widehat{v}_{mc}| \int e^{-\eta_{\phi} [|x - \bm r_{\ell}| + |x - \bm r_{k}|] - \eta_\vb |x - \bm r_m|} 
\lesssim \|\widehat{v}\|_{\ell^\infty} e^{-\frac{1}{2} \eta_\phi r_{\ell k}}$. Here, the final inequality follows from \cref{lem:elementary}.\end{proof}
 
For all $x \in \mathbb R^3$, and $\eta >0$, we have
\begin{align}
    \int \frac
        {e^{-\eta[ |\bm r_\ell - y| + |y - \bm r_k|]}}
        {|x - y|} 
    \mathrm{d}y
    &\lesssim 
    \frac
        {e^{-\frac{1}{2}\eta \, r_{\ell k}}}
        {|\bm r_\ell - x| + |x - \bm r_k|}
    \label{eq:phi*|.|-1} 
    %
    %
    %
    %
\end{align}
\begin{proof}For fixed $r > 0$, the left hand side of \cref{eq:phi*|.|-1} may be bounded above by
\begin{align}
    &e^{-\frac{1}{2} \eta \, r_{\ell k}} \bigg[ 
        \int_{B_r(x)}
            \frac
            {e^{-\frac{1}{2}\eta [|y - \bm r_\ell| + |y - \bm r_k|]}}
            {|x - y|} 
        \mathrm{d}y +
        \frac{1}{r} \int_{\mathbb R^3 \setminus B_r(x)} e^{-\frac{1}{2}\eta [|y - \bm r_\ell| + |y - \bm r_k|]} \mathrm{d}y \bigg] \nonumber \\
        &\qquad\lesssim e^{-\frac{1}{2}\eta \, r_{\ell k} }\bigg[ r^2
        e^{-\frac{1}{2}\eta [\dist( \bm r_\ell, B_{r}(x) ) + \dist( \bm r_k, B_r(x) ) ]}
        + r^{-1} \bigg].
    \nonumber
\end{align}
We conclude by choosing $r \coloneqq \frac{1}{4}[|\bm r_\ell - x| + |x - \bm r_k|]$.\end{proof}

Finally, we have
\begin{align}
    \big| V_{\mathrm{eff}}[\rho](x) \big| + e^{\theta N} \big| V_{\mathrm{eff}}[\rho](x) -  V_{\mathrm{eff}}[\rho_N](x)\big| &\lesssim 
     \sum_{m} \frac{1}{|x - \bm r_m|}
    \label{eq:Veff[rho]} 
\end{align}
\begin{proof}First, we note that $|\rho(x)| 
\lesssim \sum_{\ell k} e^{-\eta_\phi [|x - \bm r_\ell| + |x - \bm r_k|]} e^{-\eta_{\mathrm{ct}} r_{\ell k}} 
\lesssim e^{-\eta_\phi \min_\ell |x - \bm r_\ell|}$. Therefore, applying \cref{eq:phi*|.|-1}, we have
\begin{align}
    \int \frac{|\rho(y)|}{|x-y|} \mathrm{d}y 
    &\lesssim \int 
    \frac
        {e^{-\eta_\phi \min_\ell |y- \bm r_\ell|} }
        {|x-y|} 
    \mathrm{d}y
    \lesssim \sum_m \int 
    \frac
        {e^{-\eta_\phi |y- \bm r_m|} }
        {|x-y|} 
    \mathrm{d}y 
    \lesssim \sum_{m} \frac{1}{|x - \bm r_m|}.
\end{align}
The $V^{\mathrm{nuc}}$ term in the effective potential also satisfies this bound. Moreover, we have $\big|\ep_\mathrm{xc}\big(\rho(x)\big)\big| \lesssim |\rho(x)|^{1/3} \lesssim 1$. Therefore, we obtain the first bound in \cref{eq:Veff[rho]}. The second follows in the same way by using \cref{lem:locality+bodyorder} and noting  
$|\rho(x)- \rho_N(x)| 
\lesssim e^{-\theta N} e^{-\eta_\phi \min_\ell |x - \bm r_\ell|}$.
\end{proof}

\begin{proof}[Proof of \oldcref{eq:c}]
    From \cref{eq:TN,eq:T}, we have 
    \begin{align}
    T(\widehat{v}) - T_N(\widehat{v}) 
    &= \Braket{ (DF_N - D F)\big( \Ham \big) 
    \int \phi \otimes \phi \, (V - V_{\mathrm{eff}}[\rho]), \int \phi \otimes \phi \, \vb } \nonumber
    \\
    &\qquad +  \Braket{ -DF_N\big( \Ham \big) \int \phi \otimes \phi \, \Big( V_{\mathrm{eff}}[\rho_N] - V_{\mathrm{eff}}[\rho] \Big), \int \phi \otimes \phi \, \vb }.
    \label{eq:T-TN-proof}
\end{align}
We simply apply \cref{cor:DjO}, together with \cref{eq:V-bound,eq:Veff[rho],eq:phi*|.|-1}, to conclude
\begin{align}
    &\big|\big[ T(\widehat{v}) - T_N(\widehat{v}) \big]_{mc}\big| \nonumber\\
    &\qquad\lesssim 
    e^{-\theta N}
    \sum_{\ell k} 
    \left( 
        \|\widehat{v}\|_{\ell^\infty} + \sum_m \frac{1}{r_{\ell m} + r_{mk}} 
    \right)
    e^{-\frac{1}{4}\min\{\eta_\mathrm{ct}, \eta_{\phi}\} r_{\ell k}} 
    e^{-\frac{1}{2}\min\{\eta_\phi, \eta_\vb\} [r_{\ell m} + r_{mk}] } \nonumber\\
    &\qquad\lesssim Me^{-\theta N}
    \label{eq:T-TN-proof1}
\end{align}
which concludes the proof of \cref{eq:c}.
\end{proof}

Now we move on to consider the derivatives of $T_N$ and $T$: first note that 
    \begin{align}
        \frac
            {\partial \rho(x)}
            {\partial \widehat{v}_{mc}}
        &= \Braket{
            DF(\Ham) \phi(x) \otimes \phi(x),
            \int \phi \otimes \phi \, \vb_{mc}
        }.
    \end{align} 
Therefore, taking derivatives of \cref{eq:TN}, we have
\begin{align}
    DT(\widehat{v}) \widehat{w}
    &\coloneqq 
    \Braket{ - D^2 F( \Ham ) 
    \Big[ \int \phi \otimes \phi \,(V -  V_{\mathrm{eff}}[\rho]) , w \Big], \int \phi \otimes \phi \, \vb} \nonumber\\
    &\quad+ \Braket{ - D F( \Ham ) w , \int \phi \otimes \phi \, \vb} \nonumber\\
    &\quad+ \int \Braket{ DF(\Ham)\phi \otimes \phi ,
    \int\phi\otimes \phi \,\vb}  \delta V_{\mathrm{eff}}[\rho] \Braket{ DF(\Ham) \phi \otimes \phi, w } \label{eq:dTN} 
\end{align}
where $\delta V_{\mathrm{eff}}[\rho] f \coloneqq v_{\rm C} f + \ep_{\mathrm{xc}}^{\prime\prime}(\rho) f$. A similar formula for $DT_N$ results by replacing $F$ with $F_N$ and $\rho$ with $\rho_N$.

\begin{proof}[Proof of \oldcref{eq:a}]
    First, applying \cref{cor:DjO}, togther with \cref{eq:phi-assumption}, we obtain
    \begin{align}
        &\Braket{ DF(\Ham) \phi(x) \otimes \phi(x), \int \phi \otimes \phi \,\vb_{mc}} \nonumber\\
        &\qquad\lesssim \sum_{\ell k} c_\phi^2 e^{-\frac{1}{2}\min\{\eta_{\mathrm{ct}}, \eta_{\phi}\} [
            |x - \bm r_\ell| + |x - \bm r_{k}|
        ]} 
        e^{-\frac{1}{2}\min\{\eta_\phi, \eta_{\vb}\} [r_{\ell m} + r_{mk}]} \nonumber \\
        &\qquad\lesssim e^{-\frac{1}{2}
            \min\{\eta_{\mathrm{ct}}, \eta_{\phi}, \eta_{\vb}\} 
            |x - \bm r_m|
        }. 
        \label{eq:DFphi(x)}
    \end{align}
    Similarly, we have $\big| \Braket{ DF(\Ham) \phi(x) \otimes \phi(x), w } \big| \lesssim \|\widehat{w}\|_{\ell^\infty} e^{-\frac{1}{4}\min\{\eta_\mathrm{ct},\eta_\phi\} \min\limits_{1\leq\ell\leq M} |x - \bm r_{\ell}|}$.
    
    Therefore, again applying \cref{cor:DjO}, together with \cref{eq:V-bound,eq:Veff[rho],eq:phi*|.|-1}, we obtain
    \begin{align}
        \big| \big[ DT_N(\widehat{v}) \widehat{w} \big]_{mc} \big|
        &\lesssim \sum_{\ell k} \Big( 
            \|\widehat{v}\|_{\ell^\infty} + \sum_p \frac{1}{r_{\ell p} + r_{pk}} 
            + 1
        \Big) \|\widehat{w}\|_{\ell^\infty} e^{-\frac{1}{4}\min\{\eta_{\mathrm{ct}},\eta_{\phi},\eta_\vb\} [ 
        r_{m\ell} + r_{\ell k} + r_{km}] } \nonumber\\
        &\qquad+ \|\widehat{w}\|_{\ell^\infty} \iint \frac{e^{-\frac{1}{2}\min\{\eta_{\mathrm{ct}}, \eta_\phi, \eta_\vb\} |x - \bm r_m|} e^{-\frac{1}{4}\min\{\eta_{\mathrm{ct}}, \eta_{\phi} \} \min_\ell |y-\bm r_\ell| } }{|x-y|} \mathrm{d}x\mathrm{d}y \nonumber\\
        &\lesssim  \|\widehat{w}\|_{\ell^\infty} \Big[ M 
        + \int \frac{ e^{-\frac{1}{4}\min\{\eta_{\mathrm{ct}}, \eta_{\phi} \} \min_\ell |y-\bm r_\ell| } }{|y-\bm r_m|} \mathrm{d}y \Big] 
        \lesssim M\|\widehat{w}\|_{\ell^\infty},
    \end{align}
    which concludes the proof.
\end{proof} 

\begin{proof}[Proof of \oldcref{eq:d}]
    Simplifying the notation by omitting $H_0 + v$ in the notation for $F$ and $F_N$, we have
    \begin{align}
    &[DT(\widehat{v})-DT_N(\widehat{v})] \widehat{w} 
    \coloneqq 
    \Braket{ (D^2F_N - D^2 F) 
    \Big[ \int \phi \otimes \phi \,(V -  V_{\mathrm{eff}}[\rho]) , w \Big], \nabla_I v} \nonumber\\
    &\quad+ \Braket{ -D^2F_N 
    \Big[ \int \phi \otimes \phi \,(V_{\mathrm{eff}}[\rho_N] - V_{\mathrm{eff}}[\rho]) , w \Big], \nabla_I v} \nonumber\\
    &\quad+ \Braket{ (DF_N- D F) w , \nabla_I v } \nonumber\\
    &\quad+ \int \Braket{ (DF-DF_N) \phi \otimes \phi, \nabla_I v} \delta V_{\mathrm{eff}}[\rho] \Braket{ DF \phi \otimes \phi, w } \nonumber \\
    &\quad+ \int \Braket{ DF_N \phi \otimes \phi, \nabla_I v} 
    \big(\ep_{\mathrm{xc}}^{\prime\prime}(\rho) - \ep_{\mathrm{xc}}^{\prime\prime}(\rho_N) \big) 
    \Braket{ DF \phi \otimes \phi, w } \nonumber \\
    &\quad+ \int \Braket{ DF_N \phi \otimes \phi, \nabla_I v} \delta V_{\mathrm{eff}}[\rho_N] \Braket{ (DF-DF_N) \phi \otimes \phi, w }. 
    \label{eq:dTN-} 
\end{align}
In particular, by following the exact same arguments to that of the proof of \cref{eq:a}, we obtain the desired estimate. More specifically, each term in \cref{eq:dTN-} may be bounded above by similar terms from \cref{eq:dTN}, but with additional factors of $e^{-\theta N}$ coming from the difference between $F$ and $F_N$ (and their derivatives) and $\rho$ and $\rho_N$.
\end{proof}

Moreover, we have
\begin{align}
    D^2T(\widehat{v})[\widehat{w_1},\widehat{w_2}]
    &=
    \Braket{ - D^3 F \Big[\int \phi \otimes \phi \, (V - V_{\mathrm{eff}}[\rho]), w_1, w_2 \Big] , \nabla_I v} \nonumber\\
    &\quad +2\Braket{ 
    - D^2 F 
    \big[ 
        w_1, 
        w_2 
    \big], 
    \nabla_I v
    } \nonumber\\
    &\quad +
    \int 
        \Braket{
            D^2F [ 
            \phi\otimes\phi, w_1
        ] , \nabla_I v
        } \delta V_{\mathrm{eff}}[\rho] \braket{
                DF \phi\otimes \phi, w_2
            }
         \nonumber \\
    &\quad +
    \int \Braket{
        D^2F [ 
            \phi\otimes\phi, w_2
        ], \nabla_I v
    }  \delta V_{\mathrm{eff}}[\rho] \braket{
                DF \phi\otimes \phi, w_1
            }
        \nonumber \\
    &\quad +
    \int \Braket{DF \phi \otimes \phi, \nabla_I v
    } \delta V_{\mathrm{eff}}[\rho] \braket{
        D^2F[w_1,w_2],\phi\otimes\phi
    } \nonumber \\
    &\quad +
    \int \Braket{DF \phi \otimes \phi , \nabla_I v
    }  \delta^2 V_{\mathrm{eff}}[\rho] \Big[ \braket{
        DF \phi\otimes\phi, w_1
    }, 
    \braket{
        DF \phi\otimes\phi, w_2
    }
    \Big] 
\end{align}
where $DF, D^2F, D^3F$ are all evaluated at $H_0 + v$ and $\delta^2 V_{\mathrm{eff}}[\rho] [f,g] \coloneqq \ep_{\mathrm{xc}}^{\prime\prime\prime}(\rho) f g$. Similarly, we have the same expression for $T_N$ when $F$ is replaced with $F_N$ and $\rho$ with $\rho_N$. 
\begin{proof}[Proof of \oldcref{eq:b}]
    Similarly to \cref{eq:DFphi(x)}, we apply \cref{cor:DjO} to obtain
    \begin{align}
        \Big| \Braket{
            D^2F[\phi\otimes\phi(x), w], ( \nabla_I v )_{mc}
        } \Big|
        &\lesssim \|\widehat{w}\|_{\ell^\infty}
        \sum_{\ell k} 
        e^{-\frac{1}{4} \min\{\eta_{\mathrm{ct}},\eta_{\phi},\eta_{\vb}\} [|x - \bm r_\ell| + r_{\ell m} + r_{mk} + |\bm r_k - x|]} \nonumber
        \\
        &\lesssim \|\widehat{w}\|_{\ell^\infty} e^{-\frac{1}{4} \min\{\eta_{\mathrm{ct}},\eta_{\phi},\eta_{\vb}\}|x - \bm r_m|} \label{eq:DFphi(x)-2}\\
        \Big| \Braket{
            D^2F[w_1, w_2], \phi(x)\otimes\phi(x)
        } \Big|
        &\lesssim 
        \|\widehat{w}_1\|_{\ell^\infty}
        \|\widehat{w}_2\|_{\ell^\infty} 
        \sum_{\ell k} e^{-\frac{1}{4}\min\{\eta_{\mathrm{ct}}, \eta_\phi\} r_{\ell k} }
        e^{-\eta_{\phi} [|x - \bm r_\ell| + |x - \bm r_{k}|]} \nonumber \\
        &\lesssim \|\widehat{w}_1\|_{\ell^\infty}
        \|\widehat{w}_2\|_{\ell^\infty} 
        e^{-\eta_\phi \min_\ell |x - \bm r_\ell|}.
        \label{eq:DFphi(x)-3}
    \end{align}
    
    Again, we apply \cref{cor:DjO}, together with \cref{eq:V-bound,eq:phi*|.|-1,eq:Veff[rho]} and \cref{eq:DFphi(x)}, \cref{eq:DFphi(x)-2}, \cref{eq:DFphi(x)-3}, we obtain
    \begin{align}
        \left| D^2T(\widehat{v})[\widehat{w}_1,\widehat{w}_2]_{mc} \right| 
        &\lesssim \sum_{\ell k} 
        \Big( 
            \|\widehat{v}\|_{\ell^\infty} + \sum_{p} \frac{1}{r_{\ell p} + r_{pk}}
            +1
        \Big)
        \|\widehat{w}_1 \|_{\ell^\infty} 
        \|\widehat{w}_2 \|_{\ell^\infty}
        e^{-\frac{1}{2}\min\{\eta_{\phi},\eta_{\vb}\} [r_{\ell m} + r_{mk}]} \nonumber\\
        &\quad + \|\widehat{w}_1\|_{\ell^\infty} 
        \|\widehat{w}_2\|_{\ell^\infty}
        \bigg\{ 
        \iint 
        \frac
            {e^{-\frac{1}{4} \eta_{\mathrm{ct},\phi,\vb} |x - \bm r_m|} 
            e^{-\eta_\phi \min_{\ell} |y-\bm r_\ell|}}
            {|x - y|} 
        \mathrm{d}x \mathrm{d}y \nonumber\\
        &\qquad\quad + 
        \int 
            e^{-\frac{1}{4} \eta_{\mathrm{ct},\phi,\vb} |x - \bm r_m| }
            e^{-\eta_\phi \min_{\ell} |x-\bm r_\ell|} 
        \big[ 
            |\ep_{\mathrm{xc}}^{\prime\prime}\big(\rho(x)\big)| + |\ep_{\mathrm{xc}}^{\prime\prime\prime}\big(\rho(x)\big)|
        \big]
        \mathrm{d}x \bigg\} \nonumber\\
        &\lesssim M \|\widehat{w}_1\|_{\ell^\infty} \|\widehat{w}_2\|_{\ell^\infty},\nonumber
    \end{align}
    which concludes the proof.
\end{proof}

\appendix

\section{Notation} 
\label{sec:notation}
Here we summarise the key notation: 
\begin{itemize}
    \item $\bm r = \{\bm r_\ell\}_{\ell =1}^M \subset \mathbb R^d$ : atomic positions,
    
    \item $\bm{r}_{\ell k} \coloneqq \bm{r}_k - \bm{r}_\ell$ and $r_{\ell k} \coloneqq |\bm{r}_{\ell k}|$ : relative atomic positions, 
   
    \item $Z = \{Z_\ell\}_{\ell = 1}^M$ : atomic species 
    
    \item $\delta_{ij}$ : Kronecker delta ($\delta_{ij} = 0$ for $i\not=j$ and $\delta_{ii} = 1$), 
    
    \item $\mathrm{Id}_n$ : $n \times n$ identity matrix, 
    
    \item $|\,\cdot\,|$ : absolute value on $\mathbb R^d$ or $\mathbb C$,
    
    \item $|\,\cdot\,|$ : Frobenius matrix norm on $\mathbb R^{n\times n}$,

    \item $a \cdot b = \sum_i a_i b_i$ : dot product of real vectors,
    
    \item $A:B = \sum_{ij} A_{ij} B_{ij}$ : Frobenius inner product of real matrices,
    
    \item $A^\mathrm{T}$ : transpose of the matrix $A$,
    
    \item $\mathrm{Tr}$ : trace of an operator,
    
    \item $f \sim g$ as $x \to x_0 \in \mathbb R\cup\{\pm\infty\}$ or $\mathbb C \cup \{\infty\}$ : there exists an open neighbourhood $N$ of $x_0$ and positive constants $c_1,c_2 > 0$ such that $c_1 g(x) \leq f(x) \leq c_2 g(x)$ for all $x \in N$,
    
    \item $C$ : generic positive constant that may change from one line to the next, independent of system size $M$,
    
    \item $f \lesssim g$ : $f \leq C g$ for some generic positive constant, independent of system size $M$,
    
    \item $\mathbb N_0 = \{ 0, 1, 2, \dots \}$ : Natural numbers including zero,
    
    \item $\delta$ : Dirac delta, distribution satisfying $\Braket{\delta, f} = f(0)$,

    \item $\|f\|_{L^\infty(X)} \coloneqq \sup_{x\in X} |f(x)|$ : sup-norm of $f$ on $X$,
    
    \item $\mathrm{dist}(z, A) \coloneqq \inf_{a\in A} |z - a|$ : distance between $z\in\mathbb C$ and the set $A\subset \mathbb C$,
    
    \item $a + b S \coloneqq \{a + bs \colon s \in S\}$ : Minkowski addition,

    \item $\dot\bigcup_i A_i$ : union of pairwise disjoint sets $A_i$, 
    
    \item $[\psi]_\ell$ : the $\ell^\text{th}$ entry of the vector $\psi$,
    
    \item $\|\psi\|_{\ell^2}\coloneqq \left(
        \sum_{k} |[\psi]_k|^2
    \right)^{1/2}$ : $\ell^2$-norm of $\psi$,
    
    \item $\mathrm{tr}\,A \coloneqq \sum_\ell A_{\ell\ell}$ : trace of matrix $A$,
    
    \item $\|A\|_\mathrm{max} \coloneqq \max_{\ell, k} |A_{\ell k}|$ : max-norm of the matrix $A$,
    
    \item $\sigma(T)$ : the spectrum of the operator $T$,
    
    \item $\sigma_\mathrm{disc}(T) \subset \sigma(T)$ : isolated eigenvalues of finite multiplicity,
    
    \item $\sigma_\mathrm{ess}(T) \coloneqq \sigma(T) \setminus \sigma_\mathrm{disc}(T)$ : essential spectrum,
    
    \item $\|T\|_{X \to Y} \coloneqq \sup_{x \in X, \|x\|_X = 1} \|Tx\|_Y$ : operator norm of $T \colon X \to Y$,
    
    \item $\nabla v$ : Jacobian of $v \colon \mathbb R^\Lambda \to \mathbb R^\Lambda$,
    
    \item $[a, b] \coloneqq \{ (1-t) a + t b \colon t \in [0,1] \}$ : closed interval between $a,b \in \mathbb R^d$ or $a,b \in \mathbb C$,
    
    \item $\int_a^b \coloneqq \int_{[a,b]}$ : integral over the interval $[a,b]$ for $a,b \in \mathbb C$,
    
    \item $\mathrm{len}(\mathscr C)$ : length of the simple closed contour $\mathscr C$,

    \item $\mathrm{supp} \, \nu$ : support of the measure $\nu$, set of all $x$ for which every open neighbourhood of $x$ has non-zero measure,
    
    \item $\mathrm{conv}\, A \coloneqq \{ t a + (1-t) b \colon a,b \in A, t \in [0,1] \}$ : convex hull of $A$,

    \item $S_n = \{ \sigma\colon\{1,\dots,n\} \to \{1,\dots,n\} \colon \sigma \text{ bijective}\}$ : symmetric group of order $n$.
\end{itemize}

\section{Machine Learned Parameterisation \texorpdfstring{$E_\ell$}{site energy} and \texorpdfstring{$\rho_\ell$}{electron density}}
\label{sec:appendix-ml-param}

In this appendix, we describe the parameterization of $E_{N,\ell}$ and $\rho_{N,\ell}$ as body-ordered functions of $\bm{u}_k$ in more detail. The construction is a natural variation of the MACE architecture~\cite{MACE2022}, and related to the equivariant ACE model~\cite{Drautz:2020}. 
An in-depth study of our proposed architecture and its generalizations goes beyond the scope of this work; our intention is only to demonstrate the significant potential of a self-consistent ML interatomic potential model. 

In the following, atoms are indexed by $i$ and $j$, rather than $\ell$ and $k$ to avoid clashes with standard notation in equivariant networks.

\subsection{Construction of Body-Ordered Atom-centered Features}

For atom $i$, let the neighbourhood of atom $i$, denoted $\mathcal{N}(i)$, be the set atoms which are within a fixed cutoff distance $r_{\mathrm{cut}}$ from atom $i$. The environment around atom $i$ is encoded by constructing a set of learnable features $\bm{h}_i$ as functions of all neighbouring atoms.

These features are denoted by $h_{i,kLM}$, where $k = 1,...,K$ indexes independent channels, and $(L,M)$ is an angular momentum tuple. The features are equivariant with respect to rotations of the structure: if the underlying set of atomic positions $\{\bm{r}_i\}$ is rotated according to $\{\bm{r}_i\} \to \{R\bm{r}_i\}$ for some rotation $R$, then the components of $\bm{h}_i$ transform according to the Wigner matrix corresponding to $R$:
\begin{align*}
    \bm{h}_{LM}^i \rightarrow \sum_{\mu^\prime} D(R)^L_{MM^\prime}\bm{h}_{LM^\prime}^i.
\end{align*}
The features $h_{i,kLM}$ are constructed as follows: firstly, the chemical element of each atom, $z_i$, is encoded by mapping each distinct element to a vector of length $K$, via a set of weights $W$:
\begin{align*}
    a_{i,k} = \sum_{z} W_{kz} \delta_{zz_i}.
\end{align*}
The electric potential descriptors $\widehat{v}_{i,nlm}$ of atom $i$ are then introduced to create initial features. This is done by first linearly mixing the radial channels to create a vector of length $K$, followed by multiplying with the species vector $a_{i,k}$:
\begin{align*}
    \tilde{v}_{i,klm} &= \sum_n W_{lkn} \widehat{v}_{i,nlm} + \delta_{l,0} c_k\\
    h^{0}_{i,klm} &= \sum_{k'} W_{lkk'} ( a_{i,k'} \tilde{v}_{i,k'lm} ).
\end{align*}

Following this, we describe the identity and relative position of each neighbour $j$ of atom $i$ by calculating the \textit{one-particle basis} for the pair of atoms:
\begin{align*}
    \phi_{ij,k\eta l_3 m_3} = \sum_{l_1l_2m_1m_2} C^{l_3m_3}_{\eta l_1m_1 l_2m_2} R_{k \eta l_1 l_2 l_3}(r_{ij}) Y_{l_1}^{m_1}(\hat{\bm{r}}_{ij}) h^{0}_{j,kl_2m_2}
\end{align*}
where $\bm r_{ij} \coloneqq \bm r_j - \bm r_i = r_{ij} \hat{\bm r}_{ij}$ with $|\hat{\bm r}_{ij}| = 1$ and $r_{ij}\geq 0$, $Y_{l}^{m}$ is a real spherical harmonic and each $R_{\eta l_1m_1 l_2m_2}(r)$ is a learnable function of the distance $r$. This operation is inherited directly from the MACE MLIP framework, and combines the length and direction of the vector $\bm{r}_{ij}$ with equivariant features $h^{0}_{j,kl_2m_2}$ on the neighbouring atom, while preserving the equivariance of the output. The index $\eta$ appears since there may be more than one way to combine a given pair of vectors with angular momentum $(l_1, l_2)$ to get an object with angular momentum $l_3$. Further discussion can be found in \cite{MACE2022, Drautz:2020, KovacsMACEeval}.

The edge descriptors $\phi_{ij,k\eta lm}$ are then combined to form many body, atom centred descriptors.  Firstly, atomic features are created by summing over the neighbours of each atom, and applying a linear map to mix between independent channels:
\begin{align*}
    A_{i,klm} = \sum_{\tilde{k},\eta} W_{k\eta \tilde{k} l} \sum_{j \in \mathcal{N}(i)} \phi_{ij,\tilde{k}\eta lm} + \sum_{k'} W_{lkk'} h^{0}_{i,k' lm}.
\end{align*}
The initial features of atom $i$ are also added into $A_{i,klm}$, since otherwise no information about atom $i$'s electric potential would be present in $A_{i,klm}$. 

Secondly, products are taken between atomic features $A_{i,klm}$ and themselves to form many body descriptors. Specifically, products of up to $\nu$ copies of $A_{i,klm}$ are formed, while controlling the behaviour of the output under rotations:
\begin{align*}
    B^\nu_{i,\eta_\nu kLM} = \sum_{\bm{lm}} \mathcal{C}^{LM}_{\eta_\nu \bm{lm}} \prod_{\xi=1}^\nu A_{i,kl_\xi m_{\xi}}
\end{align*}
where $\mathcal{C}^{LM}_{\eta_\nu \bm{lm}}$ is the `generalised' Clebsch-Gordan coefficient, which combines a set of $\nu$ copies of $A_{i,klm}$ and generates an output which transforms in the same way as the $Y_L^M$ spherical harmonic~\cite{MACE2022,Drautz:2020}. The number of terms in the product is the number of neighbouring atoms which the features $B^\nu$ simultaneously depend on in a non-trivial way. Therefore, $\nu$ controls the \textit{body-order} of the features, with $B^\nu$ having body-order $\nu + 1$, with the $1$ coming from the the central atom.

The index $\eta_\nu$ appears because for a given set of $\{(l_\xi, m_\xi)\}_\xi$, there can be multiple ways to form a product with a given behaviour under rotations. As with the one-particle basis, this many body product has been discussed and analysed previously \cite{MACE2022, KovacsMACEeval}.

The final features describing the geometry around each atom are a linear mapping of these many body objects, summing over product order $\nu$:
\begin{align*}
    h_{i,kLM} = \sum_{\tilde{k}} W_{kL,\tilde{k}} \ \sum_\nu \sum_{\eta_\nu} \tilde{W}^\nu_{z_i \eta_\nu \tilde{k} L} B^\nu_{i,\eta_\nu \tilde{k}LM}
\end{align*}
where $W$ and $\tilde{W}$ are again learnable weights. In a MACE model, the above process is repeated, with the features $h_{i,kLM}$ replacing the chemical elements $a_{i,k}$, to iteratively construct a richer description of the geometry. In this study we do not iterate, which is equivalent to using a single ``layer'' MACE model. The motivation for using only single layer is that when using multiple layers, the features on atom $i$ can depend on more distant atoms since each layer distributes information locally. In this case, we would rather that information about distant atoms is communicated only through the electric potential.

\subsection{Parameterising \texorpdfstring{$E$}{the energy}}

Given features $h_{i,kLM}$, the energy $E_i$ is predicted by applying a on-layer multi-layer perceptron (MLP) to the invariant parts of $h_{i,kLM}$. Explicitly:
\begin{align*}
    E_i(\bm{r},\hat{v}_i) = \mathcal{R}(\{h_{i,k00}\}_k) = \sum_k W_k \ \psi \left( \sum_k' \tilde{W}_{kk'} h_{i,k00}\right).
\end{align*}
Where $\psi$ is an activation function, which in this case was a sigmoid linear unit (SiLU): $\psi(x) = \frac{x}{1 + e^{-x}}$. If $\nu$ is fixed during feature construction, $E_i$ will have body-order $\nu + 1$.

\subsection{Parameterising \texorpdfstring{$\rho$}{the density}}

The electron density is expanded in an atom centred basis, and the model must predict the coefficients of this expansion. Specifically, let $\rho_{i,lm}(\bm{r}, \hat{v})$ denote the $(l,m)$ component of the density expansion on atom $i$. This is parameterised as a linear map of the atomic features:
\begin{align*}
    \rho_{i,lm}(\bm{r}, \hat{v}) = \sum_k W_{lk} h_{i,klm},
\end{align*} 
which has body-order $\nu + 1$.

\subsection{Hyperparameters}

For the demonstrations in this study, all models used an embedding size of $K=128$, a neighbourhood cutoff of 4\AA, and atomic features $h_{i,kLM}$ were constructed for $L\leq 2$. 

\section{Charge Density Partitioning in FHI-aims}
\label{sec:appendix-fhi-aims}

In FHI-aims, the electron density is partitioned onto atom centered contributions via partitioning functions $p_\ell(\bm{r})$ which satisfy $\sum_\ell p_{\ell}(\bm{r})=1$. This allows one to coarse grain the electron density $n(\bm{r})$ into atomic multipole moments:
\begin{align*}
    \rho_{\ell,\lambda\mu} = \sqrt{\frac{4\pi}{2\lambda+1}}\int |\bm{r}-\bm{r}_\ell|^{\lambda} Y_{\lambda\mu}(\bm{r} - \bm{r}_\ell) \cdot p_{\ell}(\bm{r}) n(\bm{r})  \mathrm{d}\bm{r}.
\end{align*}
FHI-aims uses these atomic multipole moments to represent the contribution to the Hartree potential from atom $\ell$ at points far from atom $\ell$. The functions $p_\ell$ are determined by an atom-centered weight function $g_\ell(\bm{r})$ whereby $p_\ell = g_\ell / (\sum_\ell g_\ell)$. In this work, the default weighting function of FHI-aims was used which is described in \cite{STRATMANN1996213}.

\bibliography{refs.bib}
\bibliographystyle{siam}

\immediate\closeout\tempfile

\ifDRAFT{
\begin{align}
    \label{eq:}
    \tag{\textcolor{magenta}{missing ref.}}
\end{align}
\section*{Comments}
\input{lists}}

\end{document}